\documentclass[11pt]{amsart}
\usepackage[a4paper,margin=1in]{geometry}
\usepackage{amsmath,amssymb,amsthm,mathtools,bm}
\usepackage{booktabs,array,enumitem,multirow}
\usepackage[natbibapa]{apacite}
\usepackage{xcolor}
\usepackage{graphicx}
\usepackage{float}
\usepackage{setspace}
\usepackage{algorithm}
\usepackage{algpseudocode}
\usepackage{microtype}
\usepackage{hyperref}
\usepackage{siunitx}

\allowdisplaybreaks
\numberwithin{equation}{section}
\hypersetup{colorlinks=true,linkcolor=blue,citecolor=blue,urlcolor=blue}
\newcommand{\R}{\mathbb{R}}
\newcommand{\E}{\mathbb{E}}
\newcommand{\Prob}{\mathbb{P}}
\newcommand{\Var}{\mathrm{Var}}
\newcommand{\tr}{\mathrm{tr}}

\newcommand{\argmax}{\mathop{\mathrm{argmax}}}
\newcommand{\argmin}{\mathop{\mathrm{argmin}}}
\newcommand{\off}{\mathrm{off}}

\newcommand{\normzero}[1]{\left\lVert #1 \right\rVert_{0}}
\newcommand{\norminf}[1]{\left\lVert #1 \right\rVert_{\infty}}
\newcommand{\normone}[1]{\left\lVert #1 \right\rVert_{1}}
\newcommand{\normtwo}[1]{\left\lVert #1 \right\rVert_2}
\newcommand{\normmax}[1]{\left\lVert #1 \right\rVert_{\max}}
\newcommand{\normF}[1]{\left\lVert #1 \right\rVert_{\mathrm F}}
\newcommand{\abs}[1]{\left| #1 \right|}
\newcommand{\bigo}[1]{\mathcal O\!\left(#1\right)}

\newcommand{\bone}{\bm{1}}
\newcommand{\mI}{\mathbf{I}}
\newcommand{\mP}{\mathbf{\Omega}}
\newcommand{\mSigma}{\mathbf{\Sigma}}
\newcommand{\mS}{\mathbf{S}}
\newcommand{\mH}{\mathbf{H}}
\newcommand{\mR}{\mathbf{R}}
\newcommand{\mL}{\mathbf{L}}
\newcommand{\mB}{\mathbf{B}}

\newcommand{\mA}{\mathbf{A}}

\newcommand{\bx}{\bm{x}}
\newcommand{\bmu}{\bm{\mu}}
\newcommand{\bpi}{\bm{\pi}}
\newcommand{\bPi}{\bm{\Pi}}
\newcommand{\bU}{\bm{U}}

\newcommand{\bzero}{\bm{0}}

\newcommand{\br}{\bm{r}}
\newcommand{\bv}{\bm{v}}

\newcommand{\cB}{\mathcal B}
\newcommand{\cA}{\mathcal A}
\newcommand{\cP}{\mathcal P}

\newcommand{\cF}{\mathcal F}
\newcommand{\cT}{\mathcal T}

\newcommand{\cM}{\mathcal M}

\newcommand{\cC}{\mathcal C}
\newcommand{\cK}{\mathcal K}

\newcommand{\Proj}{\operatorname{Proj}}
\newcommand{\ProjPD}{\Proj_{\mathrm{PD}}}
\newcommand{\LinInterp}{\operatorname{LinInterp}}
\newcommand{\clip}{\operatorname{clip}}
\newcommand{\Tsoft}{\mathcal T}
\newcommand{\Normalize}{\mathcal N}
\newcommand{\Sphere}{\mathbb S^{p-1}}

\newtheorem{theorem}{Theorem}[section]
\newtheorem{lemma}[theorem]{Lemma}
\newtheorem{proposition}[theorem]{Proposition}
\newtheorem{corollary}[theorem]{Corollary}
\theoremstyle{definition}
\newtheorem{assumption}{Assumption}[section]

\theoremstyle{remark}

\title[Semiparametric Elliptical Mixture Clustering]{Semiparametric Elliptical Mixture Clustering for High-Dimensional Data}
\author[Long Feng and Dan Zhuang]{Long Feng$^1$ and Dan Zhuang$^2$\\ $^1$Nankai University, $^2$Fujian Normal University}
\date{}

\begin{document}

\begin{abstract}
Clustering high-dimensional data is especially challenging when cluster distributions are heavy tailed and only approximately elliptical.
Existing high-dimensional methods are largely built for Gaussian or other light-tailed models, whereas classical robust elliptical procedures are mostly low dimensional or rely on fully parametric radial families.
We propose a semiparametric elliptical mixture clustering framework with cluster-specific centers, an unknown common radial generator, and a common sparse precision-shape matrix, together with a data-driven rule for selecting the number of clusters. 
A generalized expectation--maximization (GEM) algorithm is developed by combining transformed-radius estimation of the radial generator, radial-score center updates, and a Tyler--POET--GLASSO update for the common precision-shape matrix. 
The method avoids specifying a parametric radial family and remains computationally feasible in high dimensions.
We establish high-dimensional consistency for the estimated model components and the excess misclustering error. Simulation studies and a handwritten-digit application demonstrate the competitive performance and robustness of the proposed method, particularly in heavy-tailed elliptical settings.
\end{abstract}

\maketitle


\section{Introduction}
Clustering is a central tool for uncovering latent heterogeneity in modern scientific data, but it becomes particularly difficult when the ambient dimension is comparable to, or much larger than, the sample size. Gene-expression \citep{Curtis2012}, proteomic \citep{Clarke2008}, neuroimaging \citep{Mwangi2014}, financial \citep{Ando03072017}, single-cell \citep{Menon2018}, and text data \citep{BASU2015149} often contain hundreds or thousands of variables, many of which are irrelevant or only weakly informative for clustering. In this regime, irrelevant coordinates may dominate distance-based criteria, covariance or shape estimation is often unstable, and high-dimensional noise can obscure weak cluster separation. These challenges are further amplified by heavy tails, non-Gaussian radial behavior, and nontrivial feature dependence, motivating clustering methods that can accommodate these features simultaneously.

This paper studies this problem through a semiparametric elliptical mixture model. Conditional on the latent class \(Z=k\), an observation \(\bm X\in\R^p\) has density
\begin{align*}
 f_k^*(\bm x)
 =
 |\mP^*|^{1/2}
 g^*\!\left\{(\bm x-\bmu_k^*)^\top\mP^*(\bm x-\bmu_k^*)\right\},
 \qquad k=1,\ldots,K.
\end{align*}
The centers \(\bmu_k^*\) are cluster specific, while the radial generator \(g^*\) and the precision-shape matrix \(\mP^*\) are common across clusters; the matrix \(\mP^*\) is assumed to have a sparse precision structure after an appropriate shape normalization. 
This specification retains the interpretable location-shape-radial structure
of elliptical models while avoiding the need to pre-specify a parametric radial
or tail form. The model is therefore parametric in the cluster locations and sparse precision-shape structure, but nonparametric in the common radial behavior.

If the population parameters were known, classification would be straightforward. Since the determinant term is common across components, the Bayes rule assigns \(\bm x\) to
\begin{align*}
 G^*(\bm x)
 =
 \argmax_{1\le k\le K}
 \left[
 \log \pi_k^*
 +
 \log g^*\!\left\{(\bm x-\bmu_k^*)^\top\mP^*(\bm x-\bmu_k^*)\right\}
 \right].
\end{align*}
In practice, however, the class labels are unobserved, the radial generator \(g^*\) is unknown, and the common precision-shape matrix must be estimated in high dimensions. This naturally suggests an expectation--maximization-type strategy, in which posterior responsibilities serve as soft labels. However, the usual Gaussian-mixture Expectation--Maximization (EM) updates are not directly applicable, because the radial generator is unspecified and conventional moment-based updates for centers and shape can be unstable under heavy-tailed radial variation. These considerations motivate a generalized Expectation--Maximization (GEM) framework tailored to the semiparametric elliptical mixture model.

High-dimensional clustering has often been approached by imposing sparsity,
screening variables, or reducing dimension before clustering. Sparse
\(K\)-means \citep{WittenTibshirani2010} and regularized \(K\)-means
\citep{SunWangFang2012}, for example, reduce the effect of irrelevant
coordinates through feature weighting or regularization. More recent sparse
prototype approaches include hard-thresholded prototype clustering
\citep{RaymaekersZamar2022} and robust sparse \(K\)-means with observation
weights \citep{BrodinovaFilzmoserOrtnerBreitenederRohm2019}. A related line of
work exploits informative low-dimensional structure through screening or
projection, including sparse Gaussian mixture theory
\citep{AzizyanSinghWasserman2013}, sparse non-spherical Gaussian clustering
\citep{AzizyanSinghWasserman2015}, IF-PCA \citep{JinWang2016}, and
rare-and-weak phase-transition analyses \citep{JinKeWang2017}. Median-based
methods, including sparse \(K\)-median \citep{WildMangasarian2004}, nonparametric
feature screening \citep{ChanHall2010}, and componentwise median approaches
under heavy tails \citep{HallTitteringtonXue2009}, are also relevant for robust
high-dimensional clustering. These methods are useful for reducing
high-dimensional noise or irrelevant coordinates, but their working geometry is
typically Euclidean, Gaussian, prototype-based, or coordinatewise. 
While effective for reducing high-dimensional noise, these methods are not designed to model semiparametric elliptical mixture structure or to learn an unknown heavy-tailed radial generator.

A closer line to the present work is model-based clustering, where mixture
likelihoods provide a probabilistic basis for posterior assignment. Variable
selection for Gaussian mixtures has been studied through model comparison
\citep{RafteryDean2006}, while penalized Gaussian mixture methods show how
likelihood-based clustering can be regularized in high-dimensional
low-sample-size settings \citep{PanShen2007}. Parsimonious Gaussian subspace
mixtures, such as HDDC, provide another route to dimension reduction within a
model-based framework \citep{BouveyronGirardSchmid2007}. Penalized Gaussian
mixtures with flexible covariance structures \citep{ZhouPanShen2009} and
Gaussian covariance-graph mixtures \citep{FopMurphyScrucca2019} further
incorporate covariance or dependence regularization. More broadly,
high-dimensional EM theory for latent-variable models
\citep{WangGuNingLiu2015} and Gaussian-mixture procedures such as CHIME
\citep{CaiMaZhang2019} show that EM-type algorithms can be analyzed
statistically in high-dimensional settings, including through the excess
misclustering error. These works provide useful algorithmic and theoretical
benchmarks for our GEM analysis. Their limitation for the present purpose is
that the working model is Gaussian or otherwise light-tailed, so heavy-tailed
elliptical observations may lead to model misspecification, unstable covariance
estimation, and posterior assignments that are sensitive to radial extremes.

A different route to robustness replaces Gaussian components with fully
parametric heavy-tailed families. For example, mixtures of common \(t\)-factor
analyzers accommodate heavier-than-Gaussian tails
\citep{BaekMcLachlan2011}, while skew-\(t\) factor analyzers allow additional
asymmetry \citep{MurrayBrowneMcNicholas2014}. Contaminated-Gaussian
factor-analyzer mixtures provide another robust parametric alternative
\citep{PunzoBlosteinMcNicholas2020}. These methods can substantially improve
robustness when the chosen family is well matched to the data. However, they still require the analyst to specify a concrete radial family in advance. In many applications the tail behavior may be unknown, data-dependent, or too irregular to be captured by a low-dimensional parametric radial form. 
Although these models allow heavier tails than Gaussian mixtures, they retain a parametric specification of the radial generator, which may be restrictive when the underlying tail behavior is unknown or data-dependent.

Elliptical symmetry offers a natural way to separate location, shape, and radial
behavior \citep{FangKotzNg1990}. 
In mixture clustering, multivariate \(t\)-mixture models provide a classical heavy-tailed alternative to Gaussian mixtures and improve robustness to outlying
or long-tailed observations \citep{PeelMcLachlan2000}, while identifiability of finite mixtures of elliptical distributions has been studied by \citet{HolzmannMunkGneiting2006}.
The \(t\)EIGEN family
\citep{AndrewsMcNicholas2012} and mixtures of multivariate power exponential
distributions \citep{DangBrowneMcNicholas2015} further illustrate the usefulness
of elliptical modeling with flexible component shapes and tail behavior.
Semiparametric elliptical clustering methods go further by allowing the radial
generator to be unspecified; for example,
\citet{TengFanChiangHuangLim2026} studied a clusterwise elliptical model with an
unknown radial generator and a common scatter matrix. These developments provide
a foundation for radial flexibility, but they do not address the
high-dimensional sparse precision-shape problem considered here. Existing
methods therefore cover different parts of the problem---high-dimensional noise
reduction, Gaussian mixture likelihoods, parametric tail robustness, or radial
flexibility---but do not provide a finite-mixture clustering framework that
simultaneously learns an unknown radial generator and estimates a
high-dimensional sparse precision-shape matrix.

To close this gap, the proposed method builds on robust high-dimensional tools that target different parts of the problem. Tyler's \(M\)-estimator is well suited to shape estimation under elliptical sampling with an unspecified radial distribution, because it depends on directions rather than radial magnitudes \citep{Tyler1987}. Principal Orthogonal complEment Thresholding (POET) provides a high-dimensional regularization step for
covariance or shape estimation in the presence of factor-type dependence by separating leading low-rank structure from a sparse idiosyncratic component
\citep{FanLiaoMincheva2013}; related robust elliptical developments include Kendall's tau-based factor estimation \citep{FanLiuWang2018} and Tyler--POET-type combinations \citep{XuMaWangFeng2025}. The graphical lasso
imposes sparsity on the inverse shape matrix and yields a tractable precision estimator in high dimensions \citep{FriedmanHastieTibshirani2008}. In this paper, these tools are not treated as separate clustering methods. They serve as building blocks corresponding respectively to heavy-tailed elliptical shape
estimation, high-dimensional dependence regularization, and sparse precision-shape estimation.

We integrate these ingredients through a semiparametric generalized EM algorithm. The E-step computes posterior responsibilities under the current elliptical-mixture fit. The transformed-radius block estimates the unknown common generator \(g\) by reducing the nonparametric component of the problem to one-dimensional smoothing of weighted transformed Mahalanobis radii. The resulting radial score is then used to update the component centers robustly, so that center estimation is driven by the fitted elliptical geometry rather than by raw Euclidean averages. Finally, a Tyler--POET--GLASSO block updates the common sparse precision-shape matrix: Tyler-type reweighting controls heavy-tailed radial effects, POET stabilizes high-dimensional dependence estimation, and graphical lasso enforces sparse inverse-shape structure. The procedure is a GEM algorithm because these blocks preserve the posterior-weight logic of EM while replacing unavailable exact maximizers with stable semiparametric and penalized updates.

The main contributions are threefold. First, we propose a high-dimensional semiparametric elliptical mixture clustering framework with cluster-specific centers, an unknown common radial generator, and a common sparse precision-shape matrix. Second, we develop an implementable GEM algorithm that combines transformed-radius estimation, radial-score center updates, and a Tyler--POET--GLASSO common-shape update within a single posterior-responsibility iteration. Third, we establish high-dimensional theory for the resulting estimator. The analysis controls the mixing proportions, centers, precision-shape matrix, radial generator, and excess misclustering error, thereby linking componentwise estimation consistency to clustering performance under heavy-tailed elliptical sampling.

The rest of the paper is organized as follows. Section~\ref{sec:model} introduces the semiparametric elliptical mixture model and the proposed generalized EM algorithm, including the gap rule for choosing the number of clusters. Section~\ref{sec:theory} states the high-dimensional theoretical results. Section~\ref{sec:simulations} reports simulation studies, and Section~\ref{sec:realdata} presents a handwritten-digit data analysis. Section~\ref{sec:discussion} concludes. The proofs and additional simulation results are provided in Appendices~\ref{sec:proofs} and~\ref{app:add_sim}, respectively.

\section*{Notation}
For a positive integer $m$, write $[m]=\{1,\ldots,m\}$.  Vectors are denoted by bold lower-case letters and matrices by bold upper-case letters.  For a vector $\bm v$, let $\normzero{\bm v}=|\{j:v_j\ne0\}|$, $\normone{\bm v}=\sum_j |v_j|$, $\normtwo{\bm v}=(\sum_j v_j^2)^{1/2}$, and $\norminf{\bm v}=\max_j |v_j|$.  For a matrix $\mathbf A=(A_{ij})$, let $\normF{\mathbf A}=\{\tr(\mathbf A^\top\mathbf A)\}^{1/2}$, $\normmax{\mathbf A}=\max_{i,j}|A_{ij}|$, and $\normone{\mathbf A}=\max_j\sum_i |A_{ij}|$.  The smallest and largest eigenvalues of a symmetric matrix are denoted by $\lambda_{\min}(\mathbf A)$ and $\lambda_{\max}(\mathbf A)$, respectively.  We write $\mathbf A\succ0$ when $\mathbf A$ is positive definite and $\mathbf A\succeq0$ when it is positive semidefinite.  The identity matrix is $\mI_p$, and $\bone$ is a vector of ones of conformable dimension.  If $S\subset[p]$, then $\bm v_S$ and $\mathbf A_{SS}$ denote the coordinate subvector and principal submatrix restricted to $S$.  The unit sphere in $\R^p$ is $\Sphere$.  For nonnegative sequences $a_n$ and $b_n$, $a_n\lesssim b_n$ means that $a_n\le Cb_n$ for a universal constant $C$, and $a_n\asymp b_n$ means both $a_n\lesssim b_n$ and $b_n\lesssim a_n$.  Expectations and probabilities are taken under the true data-generating distribution unless otherwise indicated.  The latent label is $Z\in[K]$; $\tau_{ik}$ denotes the posterior responsibility of observation $i$ for component $k$; and $\delta_k(\bm x)=(\bm x-\bmu_k)^\top\mP(\bm x-\bmu_k)$ denotes a squared Mahalanobis radius.

\section{Model and Method}
\label{sec:model}

\subsection{Semiparametric elliptical mixture with fixed \texorpdfstring{$K$}{K}}

Let $Z\in\{1,\dots,K\}$ be a latent class label with
\[
\Prob(Z=k)=\pi_k^*,
\qquad
\sum_{k=1}^K\pi_k^*=1,
\qquad
\pi_k^*>0.
\]
Conditional on $Z=k$, the observation $\bm X\in\R^p$ has density
\[
 f_k^*(\bm x)
 =
 |\mP^*|^{1/2}
 g^*\!\left(\delta_k^*(\bm x)\right),
 \qquad
 \delta_k^*(\bm x)
 =
 (\bm x-\bm\mu_k^*)^\top\mP^*(\bm x-\bm\mu_k^*),
\]
where $\mP^*\succ 0$ is a common precision-shape matrix, and $g^*:[0,\infty)\to[0,\infty)$ is an unknown common radial generator, and satisfying
$\frac{\pi^{p/2}}{\Gamma(p/2)}
\int_0^\infty u^{p/2-1}g^*(u)\,du=1$.
The observed marginal density is therefore
\[
 f^*(\bm x)
 =
 \sum_{k=1}^K \pi_k^* |\mP^*|^{1/2} g^*\!\left(\delta_k^*(\bm x)\right).
\]
The parameter is
\[
 z^*=(\bm\pi^*,\bm\mu_1^*,\dots,\bm\mu_K^*,\mP^*,g^*).
\]
Because the generator is unknown, $(\mP^*,g^*)$ is identifiable only up to scale. Throughout the paper we impose the trace normalization
\begin{equation}
\label{eq:shape_norm}
 \tr\{(\mP^*)^{-1}\}=p.
\end{equation}
Write
\[
 \ell^*(u)=\log g^*(u),
 \qquad
 \omega^*(u)= -\frac{d}{du}\ell^*(u).
\]
Since the determinant term is common to every component, the oracle Bayes rule is
\[
 G^*(\bm x)
 =
 \argmax_{1\le k\le K}
 \Bigl\{\log \pi_k^* + \ell^*\!\left(\delta_k^*(\bm x)\right)\Bigr\}.
\]
Under the $0$--$1$ loss, the corresponding Bayes risk is denoted by $R(G^*)=\Prob\{G^*(\bm X)\ne Z\}$.

\subsection{Oracle EM map}
The class labels are latent, so an EM-type construction is the natural starting point for model-based clustering under the mixture model.  For a given candidate parameter, the E-step replaces the unobserved labels by posterior responsibilities, and the M-step updates the model blocks conditionally on those responsibilities.  In the present semiparametric elliptical setting, however, the generator $g$ is unknown and the common precision matrix must be regularized when $p$ is large, so a literal closed-form EM maximization is neither available nor numerically stable.  We therefore define an oracle update map that preserves the exact posterior-weight structure of EM and couples it with a robust high-dimensional common-shape routine.  The oracle map below follows the same block order as the practical algorithm.
For a generic candidate parameter
\[
 z=(\bm\pi,\bm\mu_1,\dots,\bm\mu_K,\mP,g),
\]
let
\[
 \delta_k(\bm x;z)
 =
 (\bm x-\bm\mu_k)^\top\mP(\bm x-\bm\mu_k),
 \qquad
 a_k(\bm x;z)=\log \pi_k + \log g\!\left(\delta_k(\bm x;z)\right).
\]
The oracle E-step computes posterior responsibilities
\[
 r_k(\bm x;z)
 =
 \frac{\pi_k g\!\left(\delta_k(\bm x;z)\right)}{\sum_{\ell=1}^K \pi_\ell g\!\left(\delta_\ell(\bm x;z)\right)}
 =
 \frac{\exp\{a_k(\bm x;z)\}}{\sum_{\ell=1}^K \exp\{a_\ell(\bm x;z)\}}.
\]
The oracle transformed radius is
\[
 q(u)=\log(1+u),
 \qquad
 Y_k(\bm x;z)=q\!\left(\delta_k(\bm x;z)\right).
\] 
Let $f_{Y,z}$ be defined by
$\int h(y)f_{Y,z}(y)\,dy
=
\sum_{k=1}^K
\E\left[
r_k(\bm X;z)
h\{q(\delta_k(\bm X;z))\}
\right]$
for all bounded measurable functions \(h\).
The inverse change of variable from the transformed scale back to the radial scale is
\[
 \cA^{-1}(f)(u)
 =
 c_p \, u^{1-p/2} \frac{f\{q(u)\}}{1+u},
 \qquad u>0,
\]
where \((1+u)^{-1}\) is the Jacobian factor associated with \(q(u)=\log(1+u)\), \(u^{1-p/2}\) converts the one-dimensional radius density back to the elliptical generator scale under the \(p\)-dimensional radial measure, and \(c_p\) is the corresponding normalizing constant. 
The oracle 
generator update is obtained by composing this inverse transform with the shape-normalization 
projection,

\[
 G(z)=\cP\bigl[\cA^{-1}\{f_{Y,z}\}\bigr],
 \qquad
 w_z(u)= -\frac{d}{du}\log G(z)(u).
\]
The oracle mixing update is
\[
 \Pi_k(z)=\E\bigl[r_k(\bm X;z)\bigr].
\]
The center proposal is
\[
 \widetilde U_k(z)
 =
 \frac{N_k(z)}{D_k(z)},
\qquad
 N_k(z)=\E\Bigl[r_k(\bm X;z) w_z\!\left(\delta_k(\bm X;z)\right)\bm X\Bigr],
\qquad
 D_k(z)=\E\Bigl[r_k(\bm X;z) w_z\!\left(\delta_k(\bm X;z)\right)\Bigr],
\]
and the oracle damped center update is
\begin{equation}
\label{eq:oracle_pi_mu}
 U_k(z)=(1-\eta_\mu)\bm\mu_k + \eta_\mu \widetilde U_k(z),
 \qquad 0<\eta_\mu\le 1.
\end{equation}
The oracle residual after the center update is
\[
 \br_k^+(\bm X;z)=\bm X-U_k(z).
\]
The population analogue of the common-shape block follows the same sequence of robust factor adjustment, Tyler reweighting, and sparse precision regularization as the empirical algorithm.  We do not maximize the raw complete-data scatter directly, because that route is unstable when $p$ is large and the tails are heavy.  Instead, we begin with a weighted bounded-radial pilot, which is the population analogue of a spatial-sign scatter matrix in the sense of \citet{VisuriKoivunenOja2000}, use POET to separate a leading factor part from an idiosyncratic remainder, run a Tyler fixed-point update, and then apply POET once more before the sparse inverse step.  The weighted spatial-sign pilot is
\[
 \mH(z)
 =
 \sum_{k=1}^K \E\left[
 r_k(\bm X;z)
 \frac{\br_k^+(\bm X;z)\br_k^+(\bm X;z)^\top}{\normtwo{\br_k^+(\bm X;z)}^2\vee \varepsilon_r}
 \right].
\]
Let $m(z)$ be the factor count selected by the eigenvalue-ratio rule used by the algorithm, and let $\lambda_u(z)$ be the corresponding POET threshold. The first POET pilot is
\[
 \mSigma_{\mathrm{pss}}(z)=\operatorname{POET}_{\lambda_u(z),m(z)}\{\mH(z)\}.
\]
Define the weighted Tyler map
\[
 \cT_z(\mSigma)
 =
 \ProjPD\Biggl[
 \frac{p}{\tr\!\Bigl\{(1-\rho_T)\widetilde{\mSigma}_z(\mSigma)+\rho_T\mI\Bigr\}}
 \Bigl\{(1-\rho_T)\widetilde{\mSigma}_z(\mSigma)+\rho_T\mI\Bigr\}
 \Biggr],
\]
where
\[
 \widetilde{\mSigma}_z(\mSigma)
 =
 p\sum_{k=1}^K \E\left[
 r_k(\bm X;z)
 \frac{\br_k^+(\bm X;z)\br_k^+(\bm X;z)^\top}{\br_k^+(\bm X;z)^\top \mSigma^{-1}\br_k^+(\bm X;z)\vee \varepsilon_r}
 \right],
\]
$\rho_T\in(0,1)$ is the Tyler ridge parameter, and $\ProjPD$ denotes projection onto the positive-definite cone. The oracle Tyler scatter $\mSigma_{\mathrm{Ty}}(z)$ is the fixed point of $\cT_z$. The second POET step is
\[
 \mSigma_{\mathrm{pt}}(z)=\operatorname{POET}_{\lambda_u(z),m(z)}\{\mSigma_{\mathrm{Ty}}(z)\}.
\]
The sparse precision proposal is
\[
 \widetilde T(z)
 =
 \argmin_{\mP\succ\bzero}
 \Bigl\{
 \tr\bigl(\mP\mSigma_{\mathrm{pt}}(z)\bigr)-\log\det(\mP)+\lambda_\Omega(z)\normone{\off(\mP)}
 \Bigr\}.
\]
The final damped and renormalized oracle precision update is
\begin{equation}
\label{eq:T_oracle}
 T(z)
 =
 \ProjPD\Bigl(\Normalize\bigl[\ProjPD\{(1-\eta_\Omega)\mP+\eta_\Omega\widetilde T(z)\}^{-1}\bigr]^{-1}\Bigr),
 \qquad 0<\eta_\Omega\le 1,
\end{equation}
where $\Normalize(\mSigma)=p\mSigma/\tr(\mSigma)$ is the trace-normalization operator.
The oracle EM map is therefore
\[
 M(z)=\bigl(\Pi_1(z),\dots,\Pi_K(z),U_1(z),\dots,U_K(z),T(z),G(z)\bigr).
\]
By construction $M(z^*)=z^*$.

\subsection{The empirical generalized EM algorithm}
\label{subsec:empirical_gem}
We now describe the empirical counterpart of the oracle map.  
The procedure is a generalized EM algorithm because the latent-label block is updated through posterior responsibilities, while the generator and common precision blocks are updated through nonparametric and penalized surrogate maps rather than by closed-form likelihood maximization.  Each outer iteration has four statistically distinct parts: an E-step for posterior weights, a one-dimensional transformed-radius update for the common generator, a radial-score update for the centers, and a robust high-dimensional common-precision update based on Tyler's $M$-equation, POET regularization, and graphical lasso.  We write every block explicitly, because the theoretical analysis in Section~\ref{sec:theory} and Appendix \ref{sec:proofs} tracks the same sequence of operators.

The current iterate is written as
\[
 z^{(t)}=\bigl(\bm\pi^{(t)},\bm\mu_1^{(t)},\dots,\bm\mu_K^{(t)},\mP^{(t)},\hat g^{(t)}\bigr),
 \qquad
 \hat\ell^{(t)}=\log\hat g^{(t)}.
\]
The algorithm uses a robust sparse $K$-median initialization, then repeats an E-step, a transformed-radius generator update, a radial-score center update, and a TME--POET--GLASSO common-precision update.  For definiteness, the numerical stability constants are denoted by
\[
 \varepsilon_r>0,
 \qquad
 \rho_T\in(0,1),
 \qquad
 \varepsilon_{\mathrm{pd}}>0,
\]
where $\varepsilon_r$ prevents division by zero in radial normalization, $\rho_T$ regularizes the Tyler iteration, and $\varepsilon_{\mathrm{pd}}$ is used in positive-definite projection.  The generator update uses a lower radius cut-off $\varepsilon_u>0$ and clips the empirical radial score to an interval $[\omega_{\min},\omega_{\max}]$.  These constants are fixed and do not affect the asymptotic rates as long as they are kept within compact admissible ranges.

\subsubsection*{Auxiliary operators and regularization maps}
Several symbols used in the algorithm denote deterministic numerical operators rather than standard probability or optimization notation.  We define them explicitly here so that every later update is mathematically unambiguous.

For any interval $[a,b]$ with $a\le b$, the clipping operator is
\[
 \clip_{[a,b]}(x)=\min\{\max(x,a),b\},
 \qquad x\in\R.
\]

Let $0<u_1<\cdots<u_M$ and let $v_1,\dots,v_M$ be real numbers.  The linear-interpolation operator $\LinInterp(\{(u_m,v_m)\}_{m=1}^M;u)$ is defined by
\[
\LinInterp(\{(u_m,v_m)\}_{m=1}^M;u)
=
\begin{cases}
 v_1, & u\le u_1,\\[4pt]
 v_m + \dfrac{u-u_m}{u_{m+1}-u_m}(v_{m+1}-v_m), & u_m<u\le u_{m+1},\ 1\le m\le M-1,\\[8pt]
 v_M, & u\ge u_M.
\end{cases}
\]
Thus the grid values are extended piecewise linearly on the interior intervals and constantly outside the grid range.

For a real symmetric matrix $\mA$, write the eigendecomposition of its symmetrized version as
\[
 \frac{\mA+\mA^\top}{2}
 =
 \sum_{j=1}^p \lambda_j(\mA)\bv_j(\mA)\bv_j(\mA)^\top,
 \qquad
 \lambda_1(\mA)\ge\cdots\ge \lambda_p(\mA).
\]
The positive-definite projection used throughout the algorithm is
\[
 \ProjPD(\mA)
 =
 \sum_{j=1}^p \{\lambda_j(\mA)\vee \varepsilon_{\mathrm{pd}}\}\bv_j(\mA)\bv_j(\mA)^\top.
\]
This is exactly the eigenvalue-thresholding rule implemented in the computation: the matrix is first symmetrized, every eigenvalue below $\varepsilon_{\mathrm{pd}}$ is replaced by $\varepsilon_{\mathrm{pd}}$, and the matrix is reconstructed from the modified spectrum.

For a matrix $\mA=(A_{jj'})$, define the off-diagonal soft-thresholding operator by
\[
 [\Tsoft_{\lambda_u}(\mA)]_{jj'}
 =
 \begin{cases}
 A_{jj}, & j=j',\\[4pt]
 \operatorname{sgn}(A_{jj'})(\abs{A_{jj'}}-\lambda_u)_+, & j\ne j'.
 \end{cases}
\]
Following the POET construction of \citet{FanLiaoMincheva2013}, the regularization map used in this paper is
\[
 \operatorname{POET}_{\lambda_u,m}(\mA)
 =
 \ProjPD\!\left(
 \sum_{j=1}^m \lambda_j(\mA)\bv_j(\mA)\bv_j(\mA)^\top
 +
 \Tsoft_{\lambda_u}\!\left[
 \mA-\sum_{j=1}^m \lambda_j(\mA)\bv_j(\mA)\bv_j(\mA)^\top
 \right]
 \right),
\]
with the convention that the leading-eigenvalue sum is zero when $m=0$.  The first term is the retained low-rank factor part, the second term is the thresholded principal orthogonal complement, and the final projection restores positive definiteness.

\begin{algorithm}[t]
\caption{The POET operator used in the common-shape block}
\label{alg:poet}
\begin{algorithmic}[1]
\State Input a symmetric matrix $\mA$, a retained factor rank $m$, a threshold $\lambda_u$, and the projection level $\varepsilon_{\mathrm{pd}}$.
\State Symmetrize $\mA\leftarrow(\mA+\mA^\top)/2$ and compute its ordered eigendecomposition $\mA=\sum_{j=1}^p d_j\bv_j\bv_j^\top$ with $d_1\ge\cdots\ge d_p$.
\State Form the retained factor part $\mL_m(\mA)=\sum_{j=1}^m d_j\bv_j\bv_j^\top$; if $m=0$, set $\mL_m(\mA)=\mathbf 0$.
\State Form the principal orthogonal complement $\mR_m(\mA)=\mA-\mL_m(\mA)$.
\State Keep the diagonal entries of $\mR_m(\mA)$ unchanged and soft-threshold only the off-diagonal entries to obtain $\Tsoft_{\lambda_u}\{\mR_m(\mA)\}$.
\State Return $\ProjPD\bigl(\mL_m(\mA)+\Tsoft_{\lambda_u}\{\mR_m(\mA)\}\bigr)$.
\end{algorithmic}
\end{algorithm}

\subsubsection*{Robust sparse $K$-median initialization}
The empirical GEM iteration is initialized by a robust sparse $K$-median criterion.  This initialization is used only to place the algorithm in a stable basin; it is not the final clustering model.  Its design is motivated by the high-dimensional median literature: componentwise medians are robust to heavy tails in the sense emphasized by \citet{HallTitteringtonXue2009}, feature selection in $K$-median clustering was introduced by \citet{WildMangasarian2004}, and related nonparametric feature-screening ideas for very high-dimensional clustering were developed by \citet{ChanHall2010}.  The hard-thresholding form below is also consistent with the recent sparse prototype-clustering literature, where exact variable exclusion is often preferred to continuous shrinkage for interpretability \citep{RaymaekersZamar2022}.

For a threshold value $\tau>0$, let $\widetilde c_{k j}$ be the current coordinate-wise median of cluster $k$ in coordinate $j$, and define
\[
 \bar c_j=\frac1K\sum_{k=1}^K \widetilde c_{k j},
 \qquad
 D_j=\sum_{k=1}^K \abs{\widetilde c_{k j}-\bar c_j},
 \qquad
 S_\tau=\{j:D_j\ge \tau\}.
\]
\(D_j\) measures the between-cluster separation of coordinate \(j\) on the median scale. If $S_\tau$ is empty, all coordinates are retained.  Given the current medians and selected coordinate set, observations are assigned by the $L_1$ rule
\[
 c_i(\tau)
 =
 \argmin_{1\le k\le K}
 \sum_{j\in S_\tau}\abs{X_{ij}-\widetilde c_{k j}}.
\]
The coordinate-wise cluster medians are then recomputed, and the assignment and median updates are alternated until convergence.  The corresponding objective is
\[
 Q_\tau=
 \sum_{i=1}^n \min_{1\le k\le K}
 \sum_{j\in S_\tau}\abs{X_{ij}-\widetilde c_{k j}}.
\]
The threshold $\tau$ is chosen from a data-dependent finite grid.  Let $\widetilde x_j$ be the overall coordinate-wise median and let $\widetilde x_{k j}$ be the fitted cluster median.  The between-cluster $L_1$ dispersion is
\[
 B_\tau
 =
 \sum_{k=1}^K n_k \sum_{j\in S_\tau}\abs{\widetilde x_{k j}-\widetilde x_j}.
\]
For each candidate $\tau$, a reference statistic is obtained by independently permuting each column of the data matrix and refitting the same sparse $K$-median criterion.  If $B_{\tau,b}^{\mathrm{perm}}$ denotes the statistic from the $b$th permuted data set, the selected threshold is
\[
 \hat\tau
 =
 \argmax_{\tau\in\mathcal T}
 \left\{
 \log B_\tau - \frac1B\sum_{b=1}^B \log B_{\tau,b}^{\mathrm{perm}}
 \right\}.
\]
The resulting fitted medians provide the initial centers $\bm\mu_k^{(0)}$, and the initial hard memberships define an $n\times K$ responsibility matrix $\hat\tau_{ik}^{(0)}\in\{0,1\}$.

\subsubsection*{Posterior responsibilities and mixing proportions}
At iteration $t$, first compute the Mahalanobis radii
\[
 \Delta_{ik}^{(t)}
 =
 (\bm X_i-\bm\mu_k^{(t)})^\top\mP^{(t)}(\bm X_i-\bm\mu_k^{(t)}),
 \qquad
 1\le i\le n,
 \quad 1\le k\le K.
\]
The posterior responsibility of component $k$ for observation $i$ is
\[
 \hat\tau_{ik}^{(t+1)}
 =
 \frac{\pi_k^{(t)} \hat g^{(t)}\!\left(\Delta_{ik}^{(t)}\right)}
 {\sum_{\ell=1}^K \pi_\ell^{(t)} \hat g^{(t)}\!\left(\Delta_{i\ell}^{(t)}\right)}.
\]
The mixing proportion update is
\[
 \pi_k^{(t+1)} = \frac1n\sum_{i=1}^n \hat\tau_{ik}^{(t+1)},
 \qquad k=1,\dots,K.
\]
For later use, define the normalized stacked weights
\[
 \widetilde w_{ik}^{(t+1)}
 =
 \frac{\hat\tau_{ik}^{(t+1)}}{\sum_{u=1}^n\sum_{v=1}^K \hat\tau_{uv}^{(t+1)}}
 =
 \frac{\hat\tau_{ik}^{(t+1)}}{n}
\]
and the effective sample size
\[
 n_{\mathrm{eff},t+1}
 =
 \frac{\bigl(\sum_{i=1}^n\sum_{k=1}^K \hat\tau_{ik}^{(t+1)}\bigr)^2}
 {\sum_{i=1}^n\sum_{k=1}^K \{\hat\tau_{ik}^{(t+1)}\}^2}
 =
 \frac{n^2}{\sum_{i=1}^n\sum_{k=1}^K \{\hat\tau_{ik}^{(t+1)}\}^2}.
\]
Since each row of the responsibility matrix sums to one and $K$ is fixed, $n\le n_{\mathrm{eff},t+1}\le Kn$.

\subsubsection*{Transformed-radius generator update}
The generator is estimated from the same radii $\Delta_{ik}^{(t)}$ and the updated responsibilities $\hat\tau_{ik}^{(t+1)}$.  Define the transformed radii
\[
 Y_{ik}^{(t+1)}=\log\bigl(1+\Delta_{ik}^{(t)}\bigr).
\]
When a deterministic bandwidth is not prescribed, the plug-in choice is
\[
 h_{t+1}
 =
 \max\left\{h_{\min},\ 1.06\,\widehat\sigma_{Y,t+1}\,n_{\mathrm{eff},t+1}^{-1/5}\right\},
\]
where $\widehat\sigma_{Y,t+1}^2$ is the weighted variance of $Y_{ik}^{(t+1)}$ under the weights $\widetilde w_{ik}^{(t+1)}$ and $h_{\min}>0$ is a numerical lower bound.  A grid $y_1^{(t+1)}<\cdots<y_M^{(t+1)}$ is placed over the empirical transformed-radius range with a bandwidth-dependent buffer, and
\[
 u_m^{(t+1)}=
 \max\{\exp(y_m^{(t+1)})-1,\varepsilon_u\},
 \qquad m=1,\dots,M.
\]
The weighted Gaussian kernel estimator of the transformed-radius density is
\begin{equation}
\label{eq:fhat}
 \widehat f_{Y,z^{(t)},h_{t+1}}(y)
 =
 \frac{1}{h_{t+1}}
 \sum_{i=1}^n\sum_{k=1}^K
 \widetilde w_{ik}^{(t+1)}
 \phi\!\left(\frac{y-Y_{ik}^{(t+1)}}{h_{t+1}}\right),
\end{equation}
where $\phi$ is the standard normal density.  The relationship between the transformed-radius density and the elliptical generator gives the raw log-generator values
\[
 \widehat \ell_{\mathrm{raw},m}^{(t+1)}
 =
 \left(1-\frac p2\right)\log u_m^{(t+1)}
 +
 \log \widehat f_{Y,z^{(t)},h_{t+1}}\!\left(y_m^{(t+1)}\right)
 -
 \log\bigl(1+u_m^{(t+1)}\bigr)
 -
 \log \widehat C_{t+1},
\]
with numerical normalizer
\[
 \widehat C_{t+1}
 =
 \int
 \frac{\widehat f_{Y,z^{(t)},h_{t+1}}\{\log(1+u)\}}{1+u}
 \,du.
\]
In computation, this integral is approximated over the radius grid
\([u_1^{(t+1)},u_M^{(t+1)}]\) by the trapezoidal rule. 
The normalizing constant is evaluated on the radius grid by the trapezoidal rule,
\[
 \widehat C_{t+1}
 =
 \sum_{m=2}^M
 \frac{u_m^{(t+1)}-u_{m-1}^{(t+1)}}{2}
 \left[
 \frac{\widehat f_{Y,z^{(t)},h_{t+1}}\!\left(y_m^{(t+1)}\right)}{1+u_m^{(t+1)}}
 +
 \frac{\widehat f_{Y,z^{(t)},h_{t+1}}\!\left(y_{m-1}^{(t+1)}\right)}{1+u_{m-1}^{(t+1)}}
 \right].
\]
Here \(\widehat C_{t+1}\) serves as a numerical normalizer on the radius scale; multiplicative constants that do not affect the subsequent radial score are absorbed into the generator normalization convention.
Let $\mathcal S_{\lambda_{\mathrm{sp}}}$ denote the cubic smoothing-spline operator on the transformed-radius scale with a fixed smoothing parameter $\lambda_{\mathrm{sp}}>0$.  The smoothed log-generator curve is
\[
 \widehat\ell^{(t+1)}(y)
 =
 \mathcal S_{\lambda_{\mathrm{sp}}}
 \bigl[\{(y_m^{(t+1)},\widehat\ell_{\mathrm{raw},m}^{(t+1)})\}_{m=1}^M\bigr](y),
\]
and its derivative on the transformed-radius scale is
\[
 \widehat d^{(t+1)}(y)=\frac{d}{dy}\widehat\ell^{(t+1)}(y).
\]
For later interpolation back on the radius scale, define the grid values
\[
 \widehat\ell_m^{(t+1)}=\widehat\ell^{(t+1)}\!\left(y_m^{(t+1)}\right),
 \qquad
 \widehat d_m^{(t+1)}=\widehat d^{(t+1)}\!\left(y_m^{(t+1)}\right),
 \qquad m=1,\dots,M.
\]
The updated log-generator is obtained by interpolation,
\begin{equation}
\label{eq:Ghat}
 \log \widehat G_{n,h_{t+1}}(z^{(t)})(u)
 =
 \LinInterp\Bigl(\{u_m^{(t+1)},\widehat \ell_m^{(t+1)}\}_{m=1}^M;u\Bigr).
\end{equation}
The radial score used in the next center update is
\[
 \widehat\omega_m^{(t+1)}
 =
 \clip_{[\omega_{\min},\omega_{\max}]}
 \left(
 -\frac{\widehat d_m^{(t+1)}}{1+u_m^{(t+1)}}
 \right),
\]
again extended to all radii by linear interpolation.

\subsubsection*{Damped radial-score center update}
Given the updated responsibilities and radial score, define the center proposal
\[
 \widetilde{\bm\mu}_k^{(t+1)}
 =
 \frac{\sum_{i=1}^n \hat\tau_{ik}^{(t+1)}\widehat\omega^{(t+1)}\!\left(\Delta_{ik}^{(t)}\right)\bm X_i}
 {\sum_{i=1}^n \hat\tau_{ik}^{(t+1)}\widehat\omega^{(t+1)}\!\left(\Delta_{ik}^{(t)}\right)}.
\]
The accepted center update is damped,
\begin{equation}
\label{eq:sample_pi_mu}
 \bm\mu_k^{(t+1)}
 =
 (1-\eta_\mu)\bm\mu_k^{(t)} + \eta_\mu\widetilde{\bm\mu}_k^{(t+1)},
 \qquad 0<\eta_\mu\le 1.
\end{equation}
The residuals for the common-shape update are then recomputed as
\[
 \br_{ik}^{(t+1)}=\bm X_i-\bm\mu_k^{(t+1)}.
\]

\subsubsection*{TME--POET--GLASSO common-precision update}
The common-shape block is designed to remain stable under heavy-tailed elliptical sampling and high dimensionality.  The residuals $\{\br_{ik}^{(t+1)}\}$ are stacked over $i$ and $k$, with weights $\hat\tau_{ik}^{(t+1)}$.  A weighted analogue of the spatial sign covariance matrix \citep{VisuriKoivunenOja2000} is used only as a bounded-radial pilot.  Its purpose is to stabilize the first POET factor adjustment and to place the subsequent Tyler fixed-point iteration inside a numerically regular local region.  The final scatter estimator is not this pilot; it is the output of the weighted Tyler $M$-equation \citep{Tyler1987}, followed by the second POET regularization \citep{FanLiaoMincheva2013, XuMaWangFeng2025}.  The pilot itself is
\begin{equation}
\label{eq:Hhat}
 \widehat{\mH}_n(z^{(t)})
 =
 \frac{\sum_{i=1}^n\sum_{k=1}^K \hat\tau_{ik}^{(t+1)}\,\br_{ik}^{(t+1)}\br_{ik}^{(t+1)\top}/\bigl(\normtwo{\br_{ik}^{(t+1)}}^2\vee\varepsilon_r\bigr)}
 {\sum_{i=1}^n\sum_{k=1}^K \hat\tau_{ik}^{(t+1)}}.
\end{equation}
This matrix is not the final scatter estimator; the final robust scatter is obtained from Tyler's equation.  Let $d_1\ge\cdots\ge d_p$ be the eigenvalues of $\widehat{\mH}_n(z^{(t)})$.  The factor count is selected by the eigenvalue-ratio rule
\[
 V_j=\sum_{\ell=j}^{p-1} d_\ell,
 \qquad
 \mathrm{GR}_j = \frac{\log(1+d_j/V_j)}{\log(1+d_{j+1}/V_{j+1})},
 \qquad
 1\le j\le M_{\mathrm{kr}}\wedge(p-2),
\]
and $\hat m_t=\argmax_j \mathrm{GR}_j$.  If not specified otherwise, the POET threshold has order
\[
 \lambda_{u,t+1}=c_u\sqrt{\frac{\log p}{n_{\mathrm{eff},t+1}}}.
\]
The first POET regularization is
\[
 \widehat{\mSigma}_{\mathrm{pss}}^{(t+1)}
 =
 \operatorname{POET}_{\lambda_{u,t+1},\hat m_t}
 \!\left(\widehat{\mH}_n(z^{(t)})\right).
\]
This is exactly the matrix map described above and summarized in Algorithm~\ref{alg:poet}.

Starting from the trace-normalized pilot
\[
 \mSigma^{(0,t+1)}=\Normalize\!\left\{\ProjPD\bigl(\widehat{\mSigma}_{\mathrm{pss}}^{(t+1)}\bigr)\right\},
 \qquad
 \Normalize(\mA)=p\mA/\tr(\mA),
\]
the weighted Tyler iteration computes, for $\ell=0,1,2,\dots$,
\[
 q_{ik}^{(\ell,t+1)}
 =
 \br_{ik}^{(t+1)\top}\bigl(\mSigma^{(\ell,t+1)}\bigr)^{-1}\br_{ik}^{(t+1)},
\]
\[
 c_{ik}^{(\ell,t+1)}
 =
 \frac{p\,\hat\tau_{ik}^{(t+1)}}{\sum_{u=1}^n\sum_{v=1}^K \hat\tau_{uv}^{(t+1)}}\cdot
 \frac{1}{q_{ik}^{(\ell,t+1)}\vee\varepsilon_r},
\]
\[
 \widetilde{\mSigma}^{(\ell+1,t+1)}
 =
 \sum_{i=1}^n\sum_{k=1}^K c_{ik}^{(\ell,t+1)}\br_{ik}^{(t+1)}\br_{ik}^{(t+1)\top}.
\]
The next iterate is
\[
 \mSigma^{(\ell+1,t+1)}
 =
 \ProjPD
 \Bigl[
 \Normalize\bigl\{(1-\rho_T)\widetilde{\mSigma}^{(\ell+1,t+1)}+\rho_T\mI\bigr\}
 \Bigr].
\]
The iteration is run until a prescribed relative Frobenius tolerance is reached, or until a fixed maximum number of inner iterations has been completed. Its output is the weighted Tyler scatter $\widehat{\mSigma}_{\mathrm{Ty}}^{(t+1)}$.  A second POET regularization yields
\[
 \widehat{\mSigma}_{\mathrm{pt}}^{(t+1)}
 =
 \operatorname{POET}_{\lambda_{u,t+1},\hat m_t}
 \!\left(\widehat{\mSigma}_{\mathrm{Ty}}^{(t+1)}\right).
\]
The sparse precision proposal is obtained from graphical lasso.  If the base penalty is not prescribed, set
\[
 \lambda_{\Omega,t+1}^{\mathrm{base}}
 =
 c_\Omega\sqrt{\frac{\log p}{n_{\mathrm{eff},t+1}}},
\]
and form a finite multiplicative grid $\Lambda_{t+1}$ around this value.  For every $\lambda\in\Lambda_{t+1}$, solve graphical lasso with input covariance $\widehat{\mSigma}_{\mathrm{pt}}^{(t+1)}$ and choose the minimizer of
\begin{equation}
\label{eq:glasso}
 \mathrm{EBIC}_{t+1}(\lambda)
 =
 -n_{\mathrm{eff},t+1}\,\ell_\lambda
 +
 \log\bigl(n_{\mathrm{eff},t+1}\bigr)\,df_\lambda
 +
 4\gamma_{\mathrm{ebic}}\log(p)\,df_\lambda,
\end{equation}
where $\ell_\lambda$ and $df_\lambda$ are the Gaussian profile log-likelihood and the number of selected off-diagonal edges.  Denote the selected precision proposal by $\widetilde{\mP}^{(t+1)}$.  The accepted precision update is
\[
 \mP^{(t+1/2)}
 =
 \ProjPD\Bigl\{(1-\eta_\Omega)\mP^{(t)} + \eta_\Omega\widetilde{\mP}^{(t+1)}\Bigr\},
 \qquad 0<\eta_\Omega\le 1,
\]
\[
 \mSigma^{(t+1)}
 =
 \Normalize\Bigl\{\ProjPD\bigl((\mP^{(t+1/2)})^{-1}\bigr)\Bigr\},
 \qquad
 \mP^{(t+1)}
 =
 \ProjPD\bigl((\mSigma^{(t+1)})^{-1}\bigr).
\]
Thus the positive-definite projection, damping, and trace normalization preserve the shape convention $\tr\{(\mP^{(t+1)})^{-1}\}=p$.

\subsubsection*{Empirical map and final classifier}
The empirical one-step map is
\[
 M_{n,h_t}(z^{(t)})
 =
 \bigl(
 \bm\pi^{(t+1)},
 \bm\mu_1^{(t+1)},\dots,\bm\mu_K^{(t+1)},
 \mP^{(t+1)},
 \widehat G_{n,h_{t+1}}(z^{(t)})
 \bigr).
\]
The algorithm iterates
\begin{equation}
\label{eq:iteration}
 z^{(t+1)}=M_{n,h_t}(z^{(t)}),
 \qquad t=0,1,\dots,T-1,
\end{equation}
until the maximum of the center change, relative Frobenius precision change, and mixing-proportion change is below a prescribed tolerance. 
At the end of the outer iterations, the generator and responsibility matrix are recomputed once using the final centers and precision matrix.  The final clustering rule is the plug-in Bayes rule
\begin{equation}
\label{eq:plugin_rule}
 \widehat G(\bm x)
 =
 \argmax_{1\le k\le K}
 \Bigl\{\log\widehat\pi_k + \log\widehat g\!\left(\widehat\delta_k(\bm x)\right)\Bigr\},
\end{equation}
where
\[
 \widehat\delta_k(\bm x)=(\bm x-\widehat{\bm\mu}_k)^\top \widehat{\mP}(\bm x-\widehat{\bm\mu}_k).
\]
If multiple random starts are used, the selected run is the one with the largest fitted semiparametric pseudo-loglikelihood
\[
 \mathcal L^{(s)}
 =
 \sum_{i=1}^n \log\left\{\sum_{k=1}^K \pi_k^{(s)}\widehat g^{(s)}\!\left(\Delta_{ik}^{(s)}\right)\right\}.
\]

\begin{algorithm}[t]
\caption{Robust sparse $K$-median initialization}
\begin{algorithmic}[1]
\State Build a finite threshold grid $\mathcal T$ from empirical quantiles of the coordinate-wise median dispersion $D_j$.
\For{each $\tau\in\mathcal T$}
    \State Run the sparse $K$-median multistart routine with assignment rule $c_i(\tau)=\argmin_k\sum_{j\in S_\tau}\abs{X_{ij}-\widetilde c_{kj}}$.
    \State Compute the between-cluster $L_1$ dispersion $B_\tau$.
    \State Refit the same criterion to column-permuted data sets and compute $B_{\tau,b}^{\mathrm{perm}}$, $b=1,\dots,B$.
    \State Record $\log B_\tau - B^{-1}\sum_b\log B_{\tau,b}^{\mathrm{perm}}$.
\EndFor
\State Select $\hat\tau$ as the maximizer of the recorded gap score.
\State Refit sparse $K$-median at $\hat\tau$ and return initial centers and hard responsibilities.
\end{algorithmic}
\end{algorithm}

\begin{algorithm}[t]
\caption{One outer iteration of the proposed semiparametric GEM algorithm}
\begin{algorithmic}[1]
\State Compute $\Delta_{ik}^{(t)}=(\bm X_i-\bm\mu_k^{(t)})^\top\mP^{(t)}(\bm X_i-\bm\mu_k^{(t)})$.
\State Update posterior responsibilities $\hat\tau_{ik}^{(t+1)}$ and mixing proportions $\pi_k^{(t+1)}$.
\State Estimate the transformed-radius density by \eqref{eq:fhat}; reconstruct the generator by the inverse radial transformation, fixed-grid smoothing, and interpolation as in \eqref{eq:Ghat}.
\State Compute the damped radial-score center update \eqref{eq:sample_pi_mu}.
\State Form the weighted spatial-sign pilot \eqref{eq:Hhat} from the damped residuals.
\State Select the factor count by the eigenvalue-ratio rule and apply POET to obtain $\widehat{\mSigma}_{\mathrm{pss}}^{(t+1)}$.
\State Run the weighted Tyler fixed-point iteration to obtain $\widehat{\mSigma}_{\mathrm{Ty}}^{(t+1)}$.
\State Apply POET again to obtain $\widehat{\mSigma}_{\mathrm{pt}}^{(t+1)}$.
\State Solve graphical lasso on the finite penalty grid and select the minimizer of the EBIC criterion \eqref{eq:glasso}.
\State Damp, project, and trace-normalize the selected precision proposal to obtain $\mP^{(t+1)}$.
\end{algorithmic}
\end{algorithm}

\subsection{Choosing the number of clusters by a gap-LSE rule}
\label{subsec:k_selection}
The preceding algorithm assumes a fixed value of $K$.  In applications $K$ is usually unknown, and we choose it by an outer gap-statistic layer applied to the fitted GEM clustering.  This layer is separate from the statistical analysis in Section~\ref{sec:theory}, which is conditional on a fixed true number of clusters, but it is the default data-driven rule used in our implementation.

Let
\[
 \cK=\{K_{\min},K_{\min}+1,\dots,K_{\max}\}
\]
be a finite grid of candidate cluster numbers.  For each $K\in\cK$, run the full GEM algorithm described above using the same tuning conventions and the same number of random starts.  Denote the fitted quantities by
\[
 \widehat z_K=\bigl(\widehat{\bpi}_K,\widehat\bmu_{1,K},\dots,\widehat\bmu_{K,K},
 \widehat\mP_K,\widehat g_K\bigr),
 \qquad
 \widehat\tau_{ik,K},
 \qquad
 \widehat C_{i,K}=\argmax_{1\le k\le K}\widehat\tau_{ik,K}.
\]
The fitted Mahalanobis radius of observation $i$ to component $k$ under the $K$-component fit is
\[
 \widehat\Delta_{ik,K}
 =
 (\bm X_i-\widehat\bmu_{k,K})^\top
 \widehat\mP_K
 (\bm X_i-\widehat\bmu_{k,K}).
\]
Because the proposed method is designed for heavy-tailed elliptical components, we do not use the raw squared Mahalanobis radius as the within-cluster loss.  Instead we use the transformed radial loss
\[
 \rho(u)=\log(1+u),
\]
which is the same transformation used in the nonparametric generator block.  The observed within-cluster dispersion is
\begin{equation}
\label{eq:gap_observed_w}
 W_K
 =
 \frac1n\sum_{i=1}^n
 \rho\bigl(\widehat\Delta_{i\widehat C_{i,K},K}\bigr).
\end{equation}
The soft responsibility analogue is
\[
 W_K^{\mathrm{soft}}
 =
 \frac1n\sum_{i=1}^n\sum_{k=1}^K
 \widehat\tau_{ik,K}\rho(\widehat\Delta_{ik,K}),
\]
and may be used interchangeably when posterior responsibilities are nearly degenerate.  In the numerical implementation reported below, the hard-label form \eqref{eq:gap_observed_w} is used.

The reference distribution in the gap statistic should destroy clustering structure while preserving the marginal scale and tail behavior of each coordinate.  We therefore use the coordinate-permutation reference common in high-dimensional sparse clustering.  For $b=1,\dots,B$, generate a reference data matrix $\mathbf X^{0,b}$ by independently permuting each column of the observed data:
\[
 X^{0,b}_{ij}=X_{\sigma_{b,j}(i),j},
 \qquad
 i=1,\dots,n,
 \quad j=1,\dots,p,
\]
where $\sigma_{b,1},\dots,\sigma_{b,p}$ are independent random permutations of $\{1,\dots,n\}$.  This construction preserves the empirical univariate distribution of every coordinate, including heavy tails, but breaks the joint cluster structure.  It is in the same spirit as the reference resampling used in the original gap statistic of \citet{TibshiraniWaltherHastie2001Gap} and in sparse clustering criteria such as \citet{WittenTibshirani2010}.

For every reference data set $\mathbf X^{0,b}$ and every $K\in\cK$, fit the same GEM algorithm and compute the corresponding reference dispersion
\[
 W^{0,b}_K
 =
 \frac1n\sum_{i=1}^n
 \rho\bigl(\widehat\Delta^{0,b}_{i\widehat C^{0,b}_{i,K},K}\bigr).
\]
Define
\[
 \overline L_K^0
 =
 \frac1B\sum_{b=1}^B \log W_K^{0,b},
 \qquad
 s_K
 =
 \sqrt{1+B^{-1}}\left\{
 \frac1{B-1}\sum_{b=1}^B
 \bigl(\log W_K^{0,b}-\overline L_K^0\bigr)^2
 \right\}^{1/2}.
\]
The gap value is
\begin{equation}
\label{eq:gap_value}
 \operatorname{Gap}(K)=\overline L_K^0-\log W_K.
\end{equation}
Large values of \eqref{eq:gap_value} indicate that the fitted $K$-cluster partition reduces radial dispersion substantially more than would be expected under the permutation reference distribution.

There are two natural rules.  The maximum-gap rule is
\[
 \widehat K_{\max}
 =
 \argmax_{K\in\cK}\operatorname{Gap}(K),
\]
with ties resolved in favor of the smaller $K$.  In high-dimensional clustering this rule can be too aggressive because the empirical gap may keep increasing when additional small artificial clusters are introduced.  We therefore use the lower-complexity one-standard-error rule, denoted by \textsc{gap-lse} in the implementation.  Write the ordered candidate grid as $K_1<\cdots<K_L$.  The selected number of clusters is
\begin{equation}
\label{eq:gap_lse}
 \widehat K_{\mathrm{LSE}}
 =
 \min\left\{K_j: 1\le j<L,
 \operatorname{Gap}(K_j)\ge
 \operatorname{Gap}(K_{j+1})-s_{K_{j+1}}
 \right\},
\end{equation}
with the convention that $\widehat K_{\mathrm{LSE}}=K_L$ if the set in \eqref{eq:gap_lse} is empty.  This is the usual one-standard-error principle: choose the smallest model whose gap is statistically indistinguishable from the improvement obtained by moving to the next larger candidate.  In our numerical experiments the \textsc{gap-lse} rule is more stable than $\widehat K_{\max}$, and it is therefore the default rule used to choose $K$.

\begin{algorithm}[t]
\caption{Gap-LSE selection of the number of clusters}
\label{alg:gap_lse}
\begin{algorithmic}[1]
\State Input data matrix $\mathbf X\in\mathbb R^{n\times p}$, candidate grid $\cK=\{K_1,\dots,K_L\}$, number of reference samples $B$, and GEM control parameters.
\For{$K\in\cK$}
  \State Fit the full GEM algorithm to $\mathbf X$ with $K$ clusters and record $\widehat z_K$, $\widehat\tau_{ik,K}$, and $\widehat C_{i,K}$.
  \State Compute $W_K$ from \eqref{eq:gap_observed_w}.
\EndFor
\For{$b=1,\dots,B$}
  \State Generate $\mathbf X^{0,b}$ by independently permuting the entries within each column of $\mathbf X$.
  \For{$K\in\cK$}
    \State Fit the same GEM algorithm to $\mathbf X^{0,b}$ with $K$ clusters.
    \State Compute $W_K^{0,b}$ using the fitted reference labels and fitted reference precision matrix.
  \EndFor
\EndFor
\State Compute $\operatorname{Gap}(K)$ and $s_K$ for every $K\in\cK$.
\State Return $\widehat K_{\mathrm{LSE}}$ from \eqref{eq:gap_lse}.  Also record $\widehat K_{\max}$ for diagnostic comparison.
\end{algorithmic}
\end{algorithm}

\section{High-dimensional theory}
\label{sec:theory}

\subsection{Oracle, smoothed, and empirical maps}
The algorithm in Section~\ref{sec:model} has three natural versions. The first is the oracle population map $M$, in which every empirical average is replaced by its expectation and the generator block uses the unsmoothed transformed-radius density. The second is the smoothed population map $M_h$, in which the generator block is still population based, but the transformed-radius density is first convolved with the Gaussian kernel of bandwidth $h$. The third is the empirical map $M_{n,h}$ defined by the empirical update in \eqref{eq:iteration}. Throughout the theory section the bandwidth sequence $h=h_n$ is treated as a deterministic input. Data-adaptive bandwidth choices can be analyzed by conditioning on the selected bandwidth or by adding a preliminary bandwidth-selection argument.

For
\[
 z=(\bpi,\bmu_1,\dots,\bmu_K,\mP,g),
 \qquad
 \widetilde z=(\widetilde{\bpi},\widetilde{\bmu}_1,\dots,\widetilde{\bmu}_K,\widetilde{\mP},\widetilde g),
\]
write $\ell=\log g$, $\widetilde\ell=\log \widetilde g$, $w=-\ell'$, and $\widetilde w=-\widetilde\ell'$. The distance used in the theory is
\begin{equation}
\label{eq:metric}
\begin{aligned}
 d(z,\widetilde z)
 &=
 \min_{\sigma\in\mathfrak S_K}
 \Biggl[
 \norminf{\bpi-\widetilde{\bpi}_{\sigma}}
 \vee
 \max_{1\le k\le K}\norminf{\bmu_k-\widetilde{\bmu}_{\sigma(k)}}
 \vee
 \normone{\mP-\widetilde{\mP}}
 \\
 &\qquad\
 \vee
 \sup_{0\le u\le U_0}\abs{\ell(u)-\widetilde\ell(u)}
 \vee
 \sup_{0\le u\le U_0}(1+u)^{1/2}\abs{w(u)-\widetilde w(u)}
 \Biggr].
\end{aligned}
\end{equation}
where $U_0>0$ is the fixed radius cut-off used in the theoretical analysis, and $\widetilde{\bpi}_{\sigma}=(\widetilde\pi_{\sigma(1)},\dots,\widetilde\pi_{\sigma(K)})$. Let
\[
 \cB(r)=\{z:d(z,z^*)\le r\}.
\]
The minimum pairwise Mahalanobis separation is
\begin{equation}
\label{eq:DeltaStar}
 \Delta_*^2
 =
 \min_{1\le k<\ell\le K}
 (\bmu_k^*-\bmu_\ell^*)^\top \mP^*(\bmu_k^*-\bmu_\ell^*).
\end{equation}
For the sparse precision matrix write
\[
 S_\Omega=\{(a,b):\Omega_{ab}^*\neq 0\},
 \qquad
 s_\Omega=\max_{1\le a\le p}\sum_{b=1}^p \bone\{(a,b)\in S_\Omega\}.
\]
The weighted Tyler, POET, and GLASSO blocks produce the precision error scale
\begin{equation}
\label{eq:zeta_def}
 \eta_n^{\mathrm{Ty}}
 =
 \sqrt{\frac{\log p}{n}},
 \qquad
 \zeta_n
 =
 s_\Omega\Biggl(\eta_n^{\mathrm{Ty}}+\rho_{n,p}+\lambda_{\Omega,n}\Biggr),
\end{equation}
where $\rho_{n,p}$ is the POET approximation error and $\lambda_{\Omega,n}$ is the GLASSO penalty level. The full one-step statistical error is
\begin{equation}
\label{eq:eps_def}
 \varepsilon_{n,h,\zeta}
 =
 h^2
 +
 \sqrt{\frac{\log n}{nh}}
 +
 \sqrt{\frac{\log p}{n}}
 +
 \zeta_n.
\end{equation}
The term $h^2$ is the smoothing bias. The term $\sqrt{\log n/(nh)}$ is the stochastic fluctuation of the one-dimensional kernel block. The term $\sqrt{\log p/n}$ comes from the mixing-proportion and center blocks. The term $\zeta_n$ is produced only by the common TME--POET--GLASSO precision block.

\subsection{Assumptions}
The next assumptions are used throughout Section~\ref{sec:theory} and Appendix~\ref{sec:proofs}. They are stated once and then reused in all subsequent theorems.

\begin{assumption}[Semiparametric elliptical mixture and separation]
\label{ass:model}
The number of components $K$ is fixed. The data follow the common-shape elliptical mixture
\[
 \Prob(Z=k)=\pi_k^*,
 \qquad
 \bm X\mid Z=k = \bmu_k^* + \mSigma^{*1/2}R\bm U,
 \qquad k=1,\dots,K,
\]
where $\bm U$ is uniformly distributed on the unit sphere $\mathbb S^{p-1}$, $R\ge 0$ is independent of $(Z,\bm U)$, and $\mP^*=(\mSigma^*)^{-1}$. The trace normalization $\tr(\mSigma^*)=p$ holds, and the radial variable satisfies
\[
 \E(1+R)<\infty.
\]
In addition,
\[
 \min_{1\le k\le K}\pi_k^*\ge c_\pi>0,
 \qquad
 0<m_\Omega\le \lambda_{\min}(\mP^*)\le \lambda_{\max}(\mP^*)\le M_\Omega<\infty,
\]
and
\[
 \max_{1\le k\le K}\norminf{\bmu_k^*}\le M_\mu,
 \qquad
 \Delta_*^2\ge \Delta_0^2.
\]
\end{assumption}

This condition specifies the fixed-$K$ common-shape elliptical mixture that motivates the entire method. The finite first radial moment is used only in the local contraction argument to control the overlap-weighted perturbation term; it still allows heavy-tailed radial laws but rules out infinite-mean radial variables in the population EM theory. The stochastic representation, trace-normalized shape matrix, and component separation are standard in elliptical mixture and semiparametric clusterwise elliptical analyses; see \citet{FangKotzNg1990}, \citet{HolzmannMunkGneiting2006}, and \citet{TengFanChiangHuangLim2026}.

\begin{assumption}[Generator regularity and score decay]
\label{ass:generator}
The true generator $g^*$ is positive, strictly decreasing, and twice continuously differentiable on $[0,U_0]$. Writing $\ell^*=\log g^*$ and $w^*=-(\ell^*)'$, there exist constants $M_w$, $L_w$, $m_w$, $u_0$, and $c_0$ such that
\begin{equation}
\label{eq:score_decay}
 \sup_{0\le u\le U_0}(1+u)^{1/2}w^*(u)\le M_w,
 \qquad
 \sup_{0\le u\le U_0}(1+u)^{3/2}\abs{w^{*\prime}(u)}\le L_w,
\end{equation}
and
\begin{equation}
\label{eq:score_lower}
 \inf_{0\le u\le u_0} w^*(u)\ge m_w,
 \qquad
 \Prob(R^2\le u_0)\ge c_0.
\end{equation}
The transformed radial variable $Y=\log(1+R^2)$ has density $f_Y^*$ on $[0,\log(1+U_0)]$, and
\begin{equation}
\label{eq:fy_smooth}
 \sup_{y\in[0,\log(1+U_0)]}\abs{f_Y^{*\prime}(y)}
 +
 \sup_{y\in[0,\log(1+U_0)]}\abs{f_Y^{*\prime\prime}(y)}
 <\infty.
\end{equation}
\end{assumption}

The smoothness of $g^*$ controls both the stability of posterior weights and the second-order bias of the one-dimensional kernel estimator for transformed radii. Comparable radial smoothness and shape constraints are commonly imposed in semiparametric elliptical likelihood arguments and in kernel-smoothed nonparametric estimation; see \citet{TengFanChiangHuangLim2026} and \citet{Feng2026HighDimElliptical}.

\begin{assumption}[Local parameter neighborhood]
\label{ass:local}
There exists $r_0>0$ such that every $z\in\cB(r_0)$ satisfies
\[
 \pi_k\ge c_\pi/2,
 \qquad
 m_\Omega/2\le \lambda_{\min}(\mP)\le \lambda_{\max}(\mP)\le 2M_\Omega,
 \qquad
 \max_{k\le K}\norminf{\bmu_k}\le 2M_\mu,
\]
and its generator $g$ obeys the same bounds as in \eqref{eq:score_decay}--\eqref{eq:score_lower}. The score function of every $z\in\cB(r_0)$ is denoted by $w_z$.
\end{assumption}

This is a local basin condition rather than a tail or moment assumption on the data. Such neighborhoods are standard in finite-mixture EM theory, where population contraction and empirical concentration are proved only near the target parameter; see \citet{BalakrishnanWainwrightYu2017} and \citet{CaiMaZhang2019}.

\begin{assumption}[Factor structure for the common scatter]
\label{ass:factor}
The common scatter matrix admits the decomposition
\[
 \mSigma^* = \mathbf B^*\mathbf B^{*\top} + \mSigma_u^*,
\]
where $\operatorname{rank}(\mathbf B^*)=m_*$ is fixed. The eigen-gap condition
\begin{equation}
\label{eq:eigengap}
 \lambda_{m_*}(\mSigma^*)-\lambda_{m_*+1}(\mSigma^*)\ge c_{\mathrm{gap}}>0
\end{equation}
holds, and the sparse remainder obeys
\[
 \normmax{\mSigma_u^*}\le M_u,
 \qquad
 \rho_{n,p}\to 0.
\]
The thresholding rule used in the two POET steps is the entrywise soft-thresholding rule. In addition, the POET map is locally stable in entrywise norm: there exist constants $\epsilon_{\mathrm P}>0$ and $C_{\mathrm P}<\infty$, independent of $(n,p)$, such that for any symmetric matrices $\mA$ and $\mB$ satisfying
\[
 \normmax{\mA-\mSigma^*}\vee \normmax{\mB-\mSigma^*}\le \epsilon_{\mathrm P},
\]
one has
\begin{equation}
\label{eq:poet_stability_assump}
 \normmax{\operatorname{POET}_{\lambda_u,m_*}(\mA)-\operatorname{POET}_{\lambda_u,m_*}(\mB)}
 \le C_{\mathrm P}\normmax{\mA-\mB}+C_{\mathrm P}\rho_{n,p}.
\end{equation}
This max-norm stability is imposed directly; the proof below does not use the crude inequality $\normtwo{\mA-\mB}\le p\normmax{\mA-\mB}$ to control POET. If the factor rank is chosen by the eigen-ratio rule in the algorithm, then
\[
 \Prob(\widehat m = m_*)\to 1.
\]
\end{assumption}

The low-rank-plus-sparse structure is exactly the regime for which POET regularizes high-dimensional scatter matrices by thresholding the principal orthogonal complement. The explicit entrywise stability condition in \eqref{eq:poet_stability_assump} is included to make the high-dimensional max-norm argument dimension-free; without such a condition, a spectral-norm Davis--Kahan step would introduce an unwanted factor of $p$. The same structure under elliptical symmetry appears in robust POET-type analyses such as \citet{FanLiaoMincheva2013}, \citet{FanLiuWang2018}, and \citet{XuMaWangFeng2025}.

\begin{assumption}[Tyler equation and local nonsingularity]
\label{ass:tyler}
For $z\in\cB(r_0)$ and $\mSigma\succ \bzero$ with $\tr(\mSigma)=p$, define
\begin{equation}
\label{eq:Psi_def}
 \Psi(\mSigma,z)
 =
 \mSigma
 -
 \Normalize\Biggl[
 (1-\rho_T)
 p\sum_{k=1}^K
 \E\left\{
 r_k(\bm X;z)
 \frac{\br_k^+(\bm X;z)\br_k^+(\bm X;z)^\top}{\br_k^+(\bm X;z)^\top \mSigma^{-1}\br_k^+(\bm X;z)\vee \varepsilon_r}
 \right\}
 +\rho_T\mI
 \Biggr],
\end{equation}
where $\br_k^+(\bm X;z)=\bm X-U_k(z)$ and $\rho_T\in(0,1)$ is the Tyler ridge constant. The equation $\Psi(\mSigma,z)=\bzero$ has a unique solution $\mSigma_{\mathrm{Ty}}(z)$ in a neighborhood of $\mSigma^*$, and the Jacobian
\[
 \mathbf J_T = D_{\mSigma}\operatorname{vec}\Psi(\mSigma^*,z^*)
\]
is nonsingular with
\[
 \sigma_{\min}(\mathbf J_T)\ge c_T>0.
\]
\end{assumption}

The assumption states that the population Tyler shape equation is locally identifiable and differentiable. This is the analogue of a nonsingular estimating-equation Jacobian in classical $M$-estimation and allows the Tyler block to be controlled through the implicit function theorem; see \citet{Tyler1987}.

\begin{assumption}[Sparse precision regularity]
\label{ass:glasso}
Let
\[
 \mathbf\Gamma^* = \mSigma^*\otimes \mSigma^*.
\]
The standard graphical-lasso incoherence and inverse-Hessian bounds hold:
\begin{equation}
\label{eq:irrep}
 \norminf{\mathbf\Gamma^*_{S_\Omega^c S_\Omega}(\mathbf\Gamma^*_{S_\Omega S_\Omega})^{-1}}
 \le 1-\alpha,
 \qquad
 \norminf{(\mathbf\Gamma^*_{S_\Omega S_\Omega})^{-1}}\le M_\Gamma,
\end{equation}
for some $\alpha\in(0,1]$ and $M_\Gamma<\infty$. The row sparsity $s_\Omega$ may depend on $(n,p)$, but
\[
 s_\Omega\lambda_{\Omega,n}\to 0,
 \qquad
 \lambda_{\Omega,n}\asymp \sqrt{\frac{\log p}{n}}.
\]
\end{assumption}

The incoherence and inverse-Hessian bounds are the standard local regularity conditions under which graphical-lasso solutions are stable in max norm and support pattern. Closely related assumptions are used in sparse inverse covariance estimation and graphical-model selection, for example by \citet{FriedmanHastieTibshirani2008} and \citet{RavikumarWainwrightRaskuttiYu2011}.

\begin{assumption}[Local empirical-process complexity]
\label{ass:entropy}
For each fixed $h\in(0,1]$, the classes generated by the local ball $\cB(r_0)$ are considered on a deterministic transformed-radius grid $y_1,\dots,y_M$. The theory allows $M=M_n$ to grow with $n$ provided $\log M_n\lesssim\log n$ and the smoothing-spline and interpolation operators on this grid sequence have uniformly bounded $L_\infty$ operator norms; a fixed grid is the special case $M_n\equiv M$. With this convention,
\[
 \cF_{\pi,k} = \{\bx\mapsto r_k(\bx;z): z\in\cB(r_0)\},
\]
\[
 \cF_{\mu,k,j,h} = \Bigl\{\bx\mapsto r_k(\bx;z)w_{h,z}\bigl(\delta_k(\bx;z)\bigr)x_j : z\in\cB(r_0)\Bigr\},
\]
\[
 \cF_{K,m,h} = \Bigl\{\bx\mapsto r_k(\bx;z)K\Bigl(\frac{y_m-q\{\delta_k(\bx;z)\}}{h}\Bigr): z\in\cB(r_0),\ 1\le k\le K,\ 1\le m\le M\Bigr\},
\]
and
\[
 \cF_{\mathrm{Ty},a,b} = \Bigl\{\bx\mapsto r_k(\bx;z)
 \frac{r_{k,a}^+(\bx;z)r_{k,b}^+(\bx;z)}{\br_k^+(\bx;z)^\top \mSigma^{-1}\br_k^+(\bx;z)\vee \varepsilon_r}: z\in\cB(r_0),\ \mSigma\in\cC_\Sigma\Bigr\},
\]
with
\[
 \cC_\Sigma=\{\mSigma\succ\bzero: \tr(\mSigma)=p,\ m_\Sigma\mI\preceq \mSigma\preceq M_\Sigma\mI\},
\]
are VC-type classes with entropy bound
\begin{equation}
\label{eq:entropy_assump}
 \sup_Q \log N\bigl(\varepsilon\|F\|_{L_2(Q)},\cF,L_2(Q)\bigr)
 \le A_0 + v_0\log(1/\varepsilon) + C_0\log p,
 \qquad 0<\varepsilon\le 1,
\end{equation}
for envelopes $F$ that are uniformly bounded in $(n,p)$. The kernel class has bounded envelope of order $h^{-1}$ and variance of order $h^{-1}$.
\end{assumption}

This empirical-process condition is used only to invoke uniform Bernstein-type bounds over the local function classes generated by the update map. The fractional Tyler class is included explicitly in the assumption: the regularized denominator in $\cF_{\mathrm{Ty},a,b}$ is bounded away from zero by $\varepsilon_r$, and the spectral bounds on $\cC_\Sigma$ give a bounded envelope, but the VC-type entropy of this matrix-inverse fractional class is not inferred from ordinary coordinatewise VC closure alone. Thus Theorem~\ref{thm:tme_glasso} should be read as conditional on this stated local entropy bound. Such VC-type entropy assumptions are standard in high-dimensional $M$-estimation and empirical-process theory; see \citet{VanderVaartWellner1996}.

\begin{assumption}[Initialization and tuning]
\label{ass:init}
The initial iterate produced by the sparse $K$-median block satisfies
\[
 d(z^{(0)},z^*)\le r_0/2
\]
with probability tending to one. The bandwidth and penalties satisfy
\begin{equation}
\label{eq:tuning_assump}
 h=h_n\downarrow 0,
 \qquad
 \frac{nh_n}{\log n}\to\infty,
 \qquad
 \lambda_{\Omega,n}\asymp \sqrt{\frac{\log p}{n}},
 \qquad
 \rho_{n,p}\lesssim \sqrt{\frac{\log p}{n}}.
\end{equation}
\end{assumption}

The initialization requirement is the usual local-start condition for nonconvex EM-type algorithms. In our implementation the sparse median initializer is used because median-based high-dimensional procedures are less sensitive to heavy tails, as emphasized by \citet{HallTitteringtonXue2009}, while the local-basin role mirrors the initialization conditions in \citet{CaiMaZhang2019}.

\subsection{Population contraction}
The first result shows that the oracle map is locally contractive. The proof is not based on an abstract Lipschitz assumption. Instead, the required local smoothness is derived from Assumptions~\ref{ass:model}--\ref{ass:tyler}.

\begin{proposition}[Local regularity of the oracle blocks]
\label{prop:local_regularity}
Suppose Assumptions~\ref{ass:model}--\ref{ass:tyler} hold. Then there exist $r_1\in(0,r_0]$ and finite constants $L_\pi$, $L_G$, $L_\mu$, $L_{\mathrm{Ty}}$, and $L_\Omega$, depending only on the constants in Assumptions~\ref{ass:model}--\ref{ass:tyler}, such that for all $z,\widetilde z\in\cB(r_1)$,
\begin{equation}
\label{eq:local_regularity}
\begin{aligned}
 &\norminf{\Pi(z)-\Pi(\widetilde z)}
 \vee
 \sup_{0\le u\le U_0}\abs{G(z)(u)-G(\widetilde z)(u)}
 \\
 &\qquad\vee
 \sup_{0\le u\le U_0}(1+u)^{1/2}\abs{w_z(u)-w_{\widetilde z}(u)}
 \le (L_\pi+L_G) d(z,\widetilde z).
\end{aligned}
\end{equation}
\begin{equation}
\label{eq:center_regularity}
 \max_{1\le k\le K}\norminf{U_k(z)-U_k(\widetilde z)}
 \le (1-\eta_\mu+\eta_\mu L_\mu)d(z,\widetilde z),
\end{equation}
and
\begin{equation}
\label{eq:precision_regularity}
 \normmax{\mSigma_{\mathrm{Ty}}(z)-\mSigma_{\mathrm{Ty}}(\widetilde z)}
 \le L_{\mathrm{Ty}}d(z,\widetilde z),
 \qquad
 \normone{T(z)-T(\widetilde z)}
 \le \eta_\Omega L_\Omega d(z,\widetilde z).
\end{equation}
\end{proposition}

\begin{theorem}[Population contraction]
\label{thm:population}
Suppose Assumptions~\ref{ass:model}--\ref{ass:tyler} hold. Then there exist constants $r\in(0,r_1]$ and $\Delta_0^*<\infty$ such that, whenever $\Delta_0\ge \Delta_0^*$ in Assumption~\ref{ass:model}, the oracle map satisfies
\begin{equation}
\label{eq:population_contraction}
 d\bigl(M(z),z^*\bigr)
 \le \kappa\, d(z,z^*),
 \qquad z\in\cB(r),
\end{equation}
for some $\kappa\in(0,1)$. More explicitly,
\begin{equation}
\label{eq:kappa_formula}
 \kappa
 =
 \max\Bigl\{
 L_\pi^\dagger,
 L_G^\dagger,
 1-\eta_\mu+\eta_\mu L_\mu^\dagger,
 \eta_\Omega L_\Omega
 \Bigr\}+C_r r,
\end{equation}
where $L_\pi^\dagger$, $L_G^\dagger$, and $L_\mu^\dagger$ decrease to $0$ as $\Delta_0\to\infty$.
\end{theorem}

\subsection{Kernel bias and empirical concentration}
The smoothed map $M_h$ differs from $M$ only through the generator block and the center block that uses the smoothed score. The next proposition isolates the deterministic smoothing bias.

\begin{proposition}[Kernel bias]
\label{prop:bias}
Suppose Assumptions~\ref{ass:model}--\ref{ass:generator} hold, and let $K$ be the Gaussian kernel used in the generator update. Then there exists $C_h<\infty$ such that
\begin{equation}
\label{eq:bias_bound}
 \sup_{z\in\cB(r)} d\bigl(M_h(z),M(z)\bigr)
 \le C_h h^2.
\end{equation}
\end{proposition}

The next theorem controls the empirical generator, mixing-proportion, and center blocks. Here $\bU_{k,h}(z)$ denotes the smoothed population center update obtained from $U_k(z)$ by replacing $w_z$ with the kernel-smoothed score $w_{h,z}$, and $\widehat{\bU}_k(z)$ denotes its empirical counterpart. The rate contains two terms. The factor $\sqrt{\log p/n}$ comes from the finite-dimensional mixing and center blocks. The factor $\sqrt{\log n/(nh)}$ comes from the one-dimensional kernel estimator on the transformed-radius scale.

\begin{theorem}[Uniform concentration of the non-precision blocks]
\label{thm:sample_nonprecision}
Suppose Assumptions~\ref{ass:model}--\ref{ass:entropy} hold. Then there exist constants $C,c>0$ such that
\begin{equation}
\label{eq:sample_nonprecision}
\begin{aligned}
 &\Prob\Biggl[
 \sup_{z\in\cB(r)}
 \Biggl\{
 \norminf{\widehat{\bpi}(z)-\bPi(z)}
 \vee
 \max_{1\le k\le K}\norminf{\widehat{\bU}_k(z)-\bU_{k,h}(z)}
 \\
 &\qquad\
 \vee
 \sup_{0\le u\le U_0}\abs{\widehat G_{n,h}(z)(u)-G_h(z)(u)}
 \vee
 \sup_{0\le u\le U_0}(1+u)^{1/2}\abs{\widehat w_{n,h}(z)(u)-w_h(z)(u)}
 \Biggr\}
 \\
 &\qquad>
 C\Biggl\{\sqrt{\frac{\log p}{n}}+\sqrt{\frac{\log n}{nh}}\Biggr\}
 \Biggr]
 \le C p^{-c}+C n^{-c}.
\end{aligned}
\end{equation}
\end{theorem}

The next theorem treats the common covariance block. The statement is now written directly for the weighted Tyler $M$-estimator and the resulting precision matrix. Here $\mP_h(z)$ denotes the population GLASSO precision output associated with the smoothed population map, and $\widehat{\mP}_h(z)$ denotes the empirical GLASSO output. The weighted spatial-sign pilot is used only inside the proof to place the Tyler iteration inside the local basin.

\begin{theorem}[Weighted TME--POET--GLASSO concentration]
\label{thm:tme_glasso}
Suppose Assumptions~\ref{ass:model}--\ref{ass:init} hold. Then there exist constants $C,c>0$ such that
\begin{equation}
\label{eq:tme_concentration}
 \Prob\left[
 \sup_{z\in\cB(r)}
 \normmax{\widehat{\mSigma}_{\mathrm{Ty}}(z)-\mSigma_{\mathrm{Ty}}(z)}
 > C\eta_n^{\mathrm{Ty}}
 \right]
 \le Cp^{-c},
\end{equation}
and
\begin{equation}
\label{eq:ptg_concentration}
 \Prob\left[
 \sup_{z\in\cB(r)}
 \normmax{\widehat{\mSigma}_{\mathrm{pt}}(z)-\mSigma_{\mathrm{pt}}(z)}
 > C\bigl\{\eta_n^{\mathrm{Ty}}+\rho_{n,p}\bigr\}
 \right]
 \le Cp^{-c}.
\end{equation}
Moreover, the common sparse precision estimator produced by the GLASSO step obeys
\begin{equation}
\label{eq:glasso_concentration}
 \Prob\left[
 \sup_{z\in\cB(r)}
 \normone{\widehat{\mP}_h(z)-\mP_h(z)}
 > C\zeta_n
 \right]
 \le Cp^{-c}.
\end{equation}
\end{theorem}

\subsection{Main recursion, parameter consistency, and misclustering error}
The next theorem combines Theorem~\ref{thm:population}, Proposition~\ref{prop:bias}, Theorem~\ref{thm:sample_nonprecision}, and Theorem~\ref{thm:tme_glasso}.

\begin{theorem}[One-step recursion and parameter consistency]
\label{thm:main_recursion}
Suppose Assumptions~\ref{ass:model}--\ref{ass:init} hold. Let
\[
 e_t=d(z^{(t)},z^*).
\]
Then there exist constants $C,c>0$ such that, for every deterministic sequence $T_n\lesssim \log n$,
\begin{equation}
\label{eq:main_recursion}
 \Prob\Biggl[
 e_{t+1}
 \le
 \kappa e_t
 + C_1 h^2
 + C_2\sqrt{\frac{\log n}{nh}}
 + C_3\sqrt{\frac{\log p}{n}}
 + C_4\zeta_n,
 \qquad 0\le t<T_n
 \Biggr]
 \ge 1-Cp^{-c}-Cn^{-c}.
\end{equation}
Consequently,
\begin{equation}
\label{eq:unrolled_rate}
 e_t
 \le
 \kappa^t e_0
 +
 \frac{C_1 h^2+C_2\sqrt{\log n/(nh)}+C_3\sqrt{\log p/n}+C_4\zeta_n}{1-\kappa},
 \qquad 0\le t\le T_n,
\end{equation}
with the same probability. If $T_n\asymp \log(1/\varepsilon_{n,h,\zeta})$, then
\begin{equation}
\label{eq:param_consistency}
 d(z^{(T_n)},z^*) = \bigo{\varepsilon_{n,h,\zeta}}
 \qquad\text{in probability.}
\end{equation}
In particular,
\begin{equation}
\label{eq:param_components}
 \norminf{\widehat{\bpi}-\bpi^*}
 +
 \max_{1\le k\le K}\norminf{\widehat{\bmu}_k-\bmu_k^*}
 +
 \normone{\widehat{\mP}-\mP^*}
 +
 \sup_{0\le u\le U_0}\abs{\widehat g(u)-g^*(u)}
 = \bigo{\varepsilon_{n,h,\zeta}}.
\end{equation}
\end{theorem}

\begin{corollary}[Canonical rate]
\label{cor:canonical_rate}
Suppose Assumptions~\ref{ass:model}--\ref{ass:init} hold and
\[
 h\asymp \left(\frac{\log n}{n}\right)^{1/5},
 \qquad
 \rho_{n,p}\lesssim \sqrt{\frac{\log p}{n}},
 \qquad
 \lambda_{\Omega,n}\asymp \sqrt{\frac{\log p}{n}}.
\]
Then
\begin{equation}
\label{eq:canonical_rate}
 d(z^{(T_n)},z^*)
 =
 \bigo{\left(\frac{\log n}{n}\right)^{2/5} + (1+s_\Omega)\sqrt{\frac{\log p}{n}}}
 \qquad\text{in probability.}
\end{equation}
\end{corollary}

For the clustering error define the oracle and plug-in scores
\[
 \eta_k^*(\bx)=\log \pi_k^* + \ell^*\bigl(\delta_k^*(\bx)\bigr),
 \qquad
 \widehat\eta_k(\bx)=\log \widehat\pi_k + \log \widehat g\bigl(\widehat\delta_k(\bx)\bigr),
\]
and the oracle margin
\begin{equation}
\label{eq:margin_def}
 \cM^*(\bx)
 =
 \eta_{G^*(\bx)}^*(\bx)-\max_{\ell\neq G^*(\bx)}\eta_\ell^*(\bx).
\end{equation}
The final assumption is a standard margin condition adapted to the multiclass elliptical rule.

\begin{assumption}[Margin condition]
\label{ass:margin}
There exist constants $C_m>0$ and $\alpha\ge 0$ such that, almost surely with respect to the radial variable $R$,
\begin{equation}
\label{eq:margin_assump}
 \Prob\{0<\cM^*(\bm X)\le t\mid R\}\le C_m t^\alpha,
 \qquad t>0.
\end{equation}
In addition, the radial variable satisfies
\begin{equation}
\label{eq:radial_moment}
 \E(1+R)^{\alpha+1}<\infty.
\end{equation}
\end{assumption}

This is a radius-conditional Tsybakov margin condition for the posterior separation between the best and second-best component labels.  It is a mild strengthening of the usual unconditional margin assumption and is convenient here because the plug-in score error is proportional to $(1+R)$ under elliptical symmetry.  Related margin conditions are standard in classification and high-dimensional mixture analysis; see \citet{Tsybakov2004} and \citet{CaiMaZhang2019}.

\begin{theorem}[Excess misclustering error]
\label{thm:risk}
Suppose Assumptions~\ref{ass:model}--\ref{ass:margin} hold. Then
\begin{equation}
\label{eq:risk_rate}
 R(\widehat G)-R(G^*)
 =
 \bigo{\varepsilon_{n,h,\zeta}^{\alpha+1}}
 \qquad\text{in probability.}
\end{equation}
In the Lipschitz-margin case $\alpha=1$,
\begin{equation}
\label{eq:risk_rate_alpha1}
 R(\widehat G)-R(G^*)
 =
 \bigo{\varepsilon_{n,h,\zeta}^{2}}
 \qquad\text{in probability.}
\end{equation}
\end{theorem}

Theorem~\ref{thm:risk} translates the high-dimensional parameter rate in
Theorem~\ref{thm:main_recursion} into a clustering-risk rate.  The quantity
$\varepsilon_{n,h,\zeta}$ controls the uniform error of the fitted posterior scores,
whereas the margin exponent $\alpha$ determines how much probability mass lies
near the oracle decision boundary.  Hence a smaller score error produces an
excess misclustering error of order $\varepsilon_{n,h,\zeta}^{\alpha+1}$ rather
than merely $\varepsilon_{n,h,\zeta}$.  In particular, under the canonical tuning
in Corollary~\ref{cor:canonical_rate},
\[
 R(\widehat G)-R(G^*)
 =
 \bigo{
 \left[
 \left(\frac{\log n}{n}\right)^{2/5}
 +(1+s_\Omega)\sqrt{\frac{\log p}{n}}
 \right]^{\alpha+1}
 },
\]
and this becomes a quadratic rate in the common Lipschitz-margin case
$\alpha=1$.

\section{Simulation studies}
\label{sec:simulations}
We examine the finite-sample performance of the proposed semiparametric GEM procedure under simulation settings aligned with the high-dimensional elliptical model considered in the paper.
The simulations have two purposes.  
First, fixing the true number of clusters at $K=3$, we compare the clustering performance of GEM with several high-dimensional and robust competitors. 
Second, we evaluate the gap-LSE rule for selecting the number of clusters when $K$ is unknown.  Unless stated otherwise, the sample size is $n=300$, and we consider two ambient dimensions, $p=100$ and $p=200$.  The common scatter matrix is the autoregressive matrix
\[
  \mSigma^*_{ab}=0.5^{|a-b|},
  \qquad 1\le a,b\le p.
\]
The component probabilities are equal.  The centers differ only in the first six coordinates:
\[
\begin{aligned}
 \bmu_1^*&=(\delta,\delta,\delta,0,0,0,0,\ldots,0)^\top,\\
 \bmu_2^*&=(-\delta,0,0,\delta,\delta,0,0,\ldots,0)^\top,\\
 \bmu_3^*&=(0,-\delta,\delta,-\delta,0,\delta,0,\ldots,0)^\top,
\end{aligned}
\]
with $\delta=1.5$.  Thus the clustering signal is sparse in the coordinate basis, whereas the within-cluster shape is correlated through $\mSigma^*$.

Let $\mS$ be the symmetric square root of $\mSigma^*$, so that $\mSigma^*=\mS\mS^\top$, and let $\bm U$ be uniformly distributed on the unit sphere $\Sphere$.  Each component is generated in the elliptical form
\[
 \bm X=\bmu_k^*+R\mS\bm U,
 \qquad k=1,2,3,
\]
where $R\ge 0$ is independent of $\bm U$.  To match the simulation design exactly, let $Q\sim\chi_p^2$, $G\sim\chi^2_{\nu_t}$, $E\sim\operatorname{Exp}(1)$, and $V\sim\operatorname{Unif}(0,1)$, all mutually independent and also independent of $\bm U$.  The four radial laws are then
\[
 R_{\mathrm{Gauss}}=Q^{1/2},
 \qquad
 R_{t}=\left\{\frac{(\nu_t-2)Q}{G}\right\}^{1/2}
 \quad\text{with }\nu_t=5,
\]
\[
 R_{\mathrm{Lap}}=E^{1/2}Q^{1/2},
 \qquad
 R_{\mathrm{Slash}}=\left(\frac{\nu_s-2}{\nu_s}\right)^{1/2}V^{-1/\nu_s}Q^{1/2}
 \quad\text{with }\nu_s=4.
\]
Accordingly, the four distributional settings reported below are Gaussian, $t_5$, Laplace, and slash with degree parameter $4$.  The multiplicative constants in the $t$ and slash cases are chosen so that the covariance matrix, whenever it exists, equals $\mSigma^*$.  This design separates the effect of tail behavior from the effect of the common shape matrix.  The oracle procedure uses the true centers and true scatter matrix and is included only as an infeasible benchmark.

We compare seven methods: oracle classification, $K$-means, sparse $K$-means, $K$-median, sparse $K$-median, CHIME, and the proposed semiparametric GEM procedure.  Sparse $K$-means is computed using the permutation selector of \citet{WittenTibshirani2010}.  Sparse $K$-median uses the robust sparse median initialization described in Section~\ref{subsec:empirical_gem}, with its sparsity threshold selected by a permutation gap criterion. 
CHIME is run in its multiclass form with a Gaussian common-covariance working model, following the high-dimensional EM principle of \citet{CaiMaZhang2019}.  The proposed GEM procedure uses the same sparse $K$-median initialization, followed by the kernel generator update, the radial-score center update, and the TME--POET--GLASSO common-precision update described in Section~\ref{subsec:empirical_gem}.  For the proposed method we use three outer random starts, a maximum of $25$ outer GEM iterations, damping parameters $\eta_\mu=\eta_\Omega=0.7$, and the eigenvalue-ratio upper bound $8$ for the factor count in the POET steps.  All numbers are averages over $100$ Monte Carlo replications.  For compact presentation, the tables use the abbreviations \emph{Sparse KM} for sparse $K$-means and \emph{Sparse KMed} for sparse $K$-median.  In the simulation and appendix tables, boldface marks the best feasible (nonoracle) entry within each row.

Let $c_i^0\in\{1,\ldots,K\}$ denote the true class label and let $\widehat c_i$ denote the estimated cluster label.  Since cluster labels are identifiable only up to permutation, we report the label-invariant clustering accuracy
\[
  \operatorname{Acc}(\widehat{\bm c},\bm c^0)
  =
  \max_{\sigma\in\mathfrak S_K}
  \frac1n\sum_{i=1}^n
  \mathbf 1\{\sigma(\widehat c_i)=c_i^0\},
\]
where $\mathfrak S_K$ is the set of all permutations of $\{1,\ldots,K\}$.  We also report the adjusted Rand index.  Let
\[
  N_{ab}=\sum_{i=1}^n \mathbf 1\{\widehat c_i=a,\ c_i^0=b\},
  \qquad
  N_{a\cdot}=\sum_{b=1}^K N_{ab},
  \qquad
  N_{\cdot b}=\sum_{a=1}^K N_{ab}.
\]
The adjusted Rand index is
\[
\operatorname{ARI}
=
\frac{
\displaystyle
\sum_{a=1}^K\sum_{b=1}^K \binom{N_{ab}}{2}
-
\frac{
\displaystyle \sum_{a=1}^K \binom{N_{a\cdot}}{2}
\sum_{b=1}^K \binom{N_{\cdot b}}{2}}
{\binom{n}{2}}
}{
\displaystyle
\frac12\left\{
\sum_{a=1}^K \binom{N_{a\cdot}}{2}
+
\sum_{b=1}^K \binom{N_{\cdot b}}{2}
\right\}
-
\frac{
\displaystyle \sum_{a=1}^K \binom{N_{a\cdot}}{2}
\sum_{b=1}^K \binom{N_{\cdot b}}{2}}
{\binom{n}{2}}
}.
\]
Accuracy measures the fraction of correctly assigned observations after the best relabeling.  The adjusted Rand index measures pairwise agreement after correcting for chance, and is also invariant to label permutations.

\subsection{Clustering performance with \texorpdfstring{$K$}{K} fixed at the truth}\label{subsec:simulation_fixed_k}

Table~\ref{tab:simulation_accuracy} reports the average accuracy and Table~\ref{tab:simulation_ari} reports the average adjusted Rand index.  The oracle benchmark is uniformly strongest, as expected.  Among feasible procedures, the proposed GEM method attains the largest average ARI in all eight design points and the largest average accuracy in seven of the eight design points.  The only exception is the Gaussian case with $p=200$, where sparse $K$-means has a slightly higher average accuracy, $0.927$ versus $0.925$, while GEM still has the larger average ARI, $0.812$ versus $0.795$.

\begin{table}[H]
\centering
\caption{Clustering accuracy under the main simulation design.}
\label{tab:simulation_accuracy}
\scriptsize
\setlength{\tabcolsep}{3.2pt}
\begin{tabular}{lccccccc}
\toprule
\multicolumn{8}{c}{\textit{$p=100$}}\\
\midrule
Distribution & Oracle & $K$-means & Sparse KM & $K$-median & Sparse KMed & CHIME & GEM \\
\midrule
Gaussian & 0.978 & 0.917 & 0.924 & 0.700 & 0.921 & 0.934 & \textbf{0.960} \\
$t_5$ & 0.974 & 0.741 & 0.745 & 0.841 & 0.932 & 0.744 & \textbf{0.966} \\
Laplace & 0.969 & 0.878 & 0.918 & 0.756 & 0.938 & 0.878 & \textbf{0.963} \\
Slash & 0.979 & 0.923 & 0.925 & 0.742 & 0.936 & 0.934 & \textbf{0.966} \\
\midrule
\multicolumn{8}{c}{\textit{$p=200$}}\\
\midrule
Distribution & Oracle & $K$-means & Sparse KM & $K$-median & Sparse KMed & CHIME & GEM \\
\midrule
Gaussian & 0.979 & 0.880 & \textbf{0.927} & 0.500 & 0.903 & 0.879 & 0.925 \\
$t_5$ & 0.973 & 0.464 & 0.485 & 0.513 & 0.904 & 0.462 & \textbf{0.962} \\
Laplace & 0.968 & 0.547 & 0.701 & 0.514 & 0.930 & 0.568 & \textbf{0.962} \\
Slash & 0.980 & 0.863 & 0.898 & 0.508 & 0.880 & 0.862 & \textbf{0.945} \\
\bottomrule
\end{tabular}
\end{table}

\begin{table}[H]
\centering
\caption{Adjusted Rand index under the main simulation design.}
\label{tab:simulation_ari}
\scriptsize
\setlength{\tabcolsep}{3.2pt}
\begin{tabular}{lccccccc}
\toprule
\multicolumn{8}{c}{\textit{$p=100$}}\\
\midrule
Distribution & Oracle & $K$-means & Sparse KM & $K$-median & Sparse KMed & CHIME & GEM \\
\midrule
Gaussian & 0.934 & 0.769 & 0.786 & 0.345 & 0.788 & 0.815 & \textbf{0.885} \\
$t_5$ & 0.924 & 0.580 & 0.580 & 0.646 & 0.815 & 0.592 & \textbf{0.900} \\
Laplace & 0.909 & 0.736 & 0.788 & 0.570 & 0.823 & 0.744 & \textbf{0.892} \\
Slash & 0.939 & 0.790 & 0.797 & 0.412 & 0.820 & 0.824 & \textbf{0.900} \\
\midrule
\multicolumn{8}{c}{\textit{$p=200$}}\\
\midrule
Distribution & Oracle & $K$-means & Sparse KM & $K$-median & Sparse KMed & CHIME & GEM \\
\midrule
Gaussian & 0.939 & 0.683 & 0.795 & 0.088 & 0.754 & 0.681 & \textbf{0.812} \\
$t_5$ & 0.921 & 0.192 & 0.220 & 0.162 & 0.765 & 0.189 & \textbf{0.889} \\
Laplace & 0.907 & 0.306 & 0.506 & 0.205 & 0.808 & 0.335 & \textbf{0.888} \\
Slash & 0.940 & 0.667 & 0.759 & 0.100 & 0.725 & 0.666 & \textbf{0.852} \\
\bottomrule
\end{tabular}
\end{table}

The advantage of GEM is most pronounced in the heavy-tailed cases.  For the $t_5$ mixture with $p=200$, GEM increases the average accuracy from $0.462$ for CHIME to $0.962$, and increases the average ARI from $0.189$ to $0.889$.  For the Laplace mixture with $p=200$, the corresponding improvements are from $0.568$ to $0.962$ in accuracy and from $0.335$ to $0.888$ in ARI.  These comparisons are consistent with the purpose of the proposed method: the posterior update is not tied to a Gaussian radial density, and the common precision update is based on Tyler-type elliptical shape estimation rather than on sample covariance estimation.

The comparison with sparse $K$-median separates robust coordinatewise initialization from the full semiparametric elliptical update.  Sparse $K$-median performs well in several heavy-tailed cases, but it does not use the common elliptical shape.  In the slash design with $p=200$, for example, GEM improves the average accuracy from $0.880$ to $0.945$ and the average ARI from $0.725$ to $0.852$.  Thus the TME--POET--GLASSO shape update and the nonparametric radial generator both contribute beyond robust feature screening alone.  In the Gaussian setting, simpler Euclidean and sparse Euclidean methods remain competitive, but GEM remains stable and close to the oracle benchmark.  Overall, the simulation supports the theoretical message of the paper: when the data are high dimensional, elliptically distributed, and potentially heavy tailed, estimating the radial generator and the common sparse precision shape can materially improve clustering accuracy.  Appendix~\ref{app:add_sim} reports two additional sensitivity analyses, one under a dense compound-symmetric common scatter matrix and one under a dense-mean design, both of which complement the main sparse-mean experiments.

\subsection{Selection of the number of clusters}
\label{subsec:simulation_gap_k}

We next examine the finite-sample behavior of the cluster-number selector in Section~\ref{subsec:k_selection}.  The simulation setting is the same as above except that $K$ is treated as unknown and the proposed GEM algorithm is fitted over the candidate grid
\[
  \cK=\{2,3,4,5\}.
\]
The true number of clusters is $K_0=3$, the sample size is $n=300$, the common AR(1) correlation parameter is $0.5$, and the signal strength is $\delta=1.5$.  We report results for two ambient dimensions, $p=100$ and $p=200$.  For each candidate $K$, the gap statistic is computed with $B=20$ coordinate-wise permutation reference samples.  We compare the one-standard-error selector $\widehat K_{\mathrm{LSE}}$ with the maximum-gap selector $\widehat K_{\max}$.  The experiment is repeated for the Gaussian, $t_5$, Laplace, and slash radial laws, using $100$ Monte Carlo replications for each distribution and dimension.

\begin{table}[H]
\centering
\caption{Gap-based selection of the number of clusters under the main simulation design.}
\label{tab:gap_k_selection}
\scriptsize
\setlength{\tabcolsep}{3.3pt}
\begin{tabular}{llccccccc}
\toprule
Distribution & Rule & $\widehat K=2$ & $\widehat K=3$ & $\widehat K=4$ & $\widehat K=5$ & True-$K$ rate & Accuracy & ARI \\
\midrule
\multicolumn{9}{c}{\textit{$p=100$}}\\
\midrule
\multirow{2}{*}{Gaussian} & Gap-LSE & 4 & 93 & 3 & 0 & \textbf{0.93} & 0.9516 & 0.8772 \\
 & Gap-max & 0 & 81 & 14 & 5 & 0.81 & 0.9615 & 0.8641 \\
\addlinespace[2pt]

\multirow{2}{*}{$t_5$} & Gap-LSE & 0 & 100 & 0 & 0 & \textbf{1.00} & 0.9656 & 0.8998 \\
 & Gap-max & 0 & 97 & 2 & 1 & 0.97 & 0.9657 & 0.8964 \\
\addlinespace[2pt]

\multirow{2}{*}{Laplace} & Gap-LSE & 0 & 100 & 0 & 0 & \textbf{1.00} & 0.9632 & 0.8930 \\
 & Gap-max & 0 & 99 & 0 & 1 & 0.99 & 0.9631 & 0.8921 \\
\addlinespace[2pt]

\multirow{2}{*}{Slash} & Gap-LSE & 1 & 98 & 1 & 0 & \textbf{0.98} & 0.9642 & 0.8997 \\
 & Gap-max & 0 & 94 & 4 & 2 & 0.94 & 0.9674 & 0.8953 \\
\midrule

\multicolumn{9}{c}{\textit{$p=200$}}\\
\midrule
\multirow{2}{*}{Gaussian} & Gap-LSE & 16 & 70 & 13 & 1 & \textbf{0.70} & 0.8912 & 0.7659 \\
 & Gap-max & 0 & 41 & 33 & 26 & 0.41 & 0.9358 & 0.7460 \\
\addlinespace[2pt]

\multirow{2}{*}{$t_5$} & Gap-LSE & 7 & 89 & 4 & 0 & \textbf{0.89} & 0.9375 & 0.8528 \\
 & Gap-max & 0 & 84 & 14 & 2 & 0.84 & 0.9609 & 0.8695 \\
\addlinespace[2pt]

\multirow{2}{*}{Laplace} & Gap-LSE & 0 & 97 & 3 & 0 & \textbf{0.97} & 0.9615 & 0.8881 \\
 & Gap-max & 0 & 92 & 4 & 4 & 0.92 & 0.9609 & 0.8837 \\
\addlinespace[2pt]

\multirow{2}{*}{Slash} & Gap-LSE & 16 & 80 & 4 & 0 & \textbf{0.80} & 0.9057 & 0.8094 \\
 & Gap-max & 2 & 78 & 18 & 2 & 0.78 & 0.9562 & 0.8618 \\
\bottomrule
\end{tabular}
\end{table}

Table~\ref{tab:gap_k_selection} shows that the gap-LSE rule is consistently more reliable for estimating $K_0$ than the maximum-gap rule.  When $p=100$, gap-LSE selects the true number of clusters in $93\%$, $100\%$, $100\%$, and $98\%$ of the replications for the Gaussian, $t_5$, Laplace, and slash designs, respectively.  When $p=200$, the selection problem becomes more difficult, but gap-LSE still attains true-$K$ rates of $70\%$, $89\%$, $97\%$, and $80\%$ across the same four radial laws.  These rates are uniformly higher than those of the maximum-gap rule, which obtains $41\%$, $84\%$, $92\%$, and $78\%$ at $p=200$.

The difference between the two rules is mainly due to the behavior of the maximum-gap rule in high dimension.  At $p=200$, the maximum-gap rule selects $K=4$ or $K=5$ in $59$ Gaussian replications, $16$ $t_5$ replications, $8$ Laplace replications, and $20$ slash replications.  The one-standard-error rule reduces this tendency to over-split by selecting the smallest candidate whose gap is statistically indistinguishable from the next larger candidate.  This conservative bias can occasionally lead to under-selection, most visibly in the Gaussian and slash designs at $p=200$, where $\widehat K_{\mathrm{LSE}}=2$ occurs in $16$ out of $100$ replications.  Nevertheless, because the goal of this subsection is cluster-number recovery rather than maximizing relabeled classification accuracy, the true-$K$ rate is the primary criterion.  The results therefore support using $\widehat K_{\mathrm{LSE}}$ as the default data-driven selector in the proposed GEM procedure.

\section{Real-data analysis}
\label{sec:realdata}

We complement the Monte Carlo evidence with an unsupervised analysis of the optical recognition of handwritten digits data set from the UCI Machine Learning Repository \citep{DuaGraff2019UCI}.  The data set consists of the standard training and test files, with $3823$ training observations and $1797$ test observations.  Each observation is a $64$-dimensional vector obtained by partitioning a normalized $32\times 32$ bitmap into $4\times 4$ nonoverlapping blocks and recording the number of active pixels in each block.  Thus each digit image is represented as an $8\times 8$ array of block counts.  We combine the training and test files, obtaining $n=5620$ observations in dimension $p=64$.  The labels, which take values in $\{0,1,\ldots,9\}$, are used only for performance evaluation.  The ten class sizes range from $554$ to $572$, and there are no missing values.

All variables are standardized by their sample means and sample standard deviations over the combined sample.  We compare the six data-driven methods used in the simulations: $K$-means, sparse $K$-means, $K$-median, sparse $K$-median, CHIME, and GEM. As in Section~\ref{sec:simulations}, the tables abbreviate sparse $K$-means and sparse $K$-median by \emph{Sparse KM} and \emph{Sparse KMed}.  Boldface marks the best entry in each row.  We consider three complementary protocols.  First, we cluster the full data set into $K=10$ groups.  Second, for each of the $\binom{10}{2}=45$ unordered digit pairs, we retain only observations from the corresponding two classes and cluster the resulting subsample into $K=2$ groups. Third, for each of the $\binom{10}{3}=120$ unordered digit triplets, we retain only observations from the corresponding three classes and cluster the resulting subsample into $K=3$ groups.  For the pairwise and triplet protocols, Table~\ref{tab:optdigits} reports averages and standard deviations across the collection of subproblems. Accuracy is computed after optimal relabeling, and ARI denotes the adjusted Rand index.

\begin{table}[H]
	\centering
	\caption{Clustering results for the Optdigits data.}
	\label{tab:optdigits}
	\scriptsize
	\setlength{\tabcolsep}{3.0pt}
	\resizebox{\textwidth}{!}{%
		\begin{tabular}{llcccccc}
			\toprule
			Protocol & Metric & $K$-means & Sparse KM & $K$-median & Sparse KMed & CHIME & GEM \\
			\midrule
			\multirow{2}{*}{Full $K=10$} 
			& Accuracy
			& 0.6399 & 0.5692 & \textbf{0.7740} & 0.7062 & 0.6230 & 0.7512 \\
			& ARI
			& 0.4910 & 0.3846 & \textbf{0.6369} & 0.5559 & 0.4601 & 0.6304 \\
			\addlinespace[2pt]
			
			\multirow{2}{*}{Pairwise $K=2$} 
			& Accuracy
			& 0.9556 (0.0781) & 0.9321 (0.1116) & 0.9567 (0.0720)
			& 0.9667 (0.0360) & 0.9663 (0.0873) & \textbf{0.9785 (0.0453)} \\
			& ARI
			& 0.8539 (0.2123) & 0.7954 (0.2815) & 0.8544 (0.1990)
			& 0.8762 (0.1223) & 0.8996 (0.2309) & \textbf{0.9239 (0.1394)} \\
			\addlinespace[2pt]
			
			\multirow{2}{*}{Triplet $K=3$} 
			& Accuracy
			& 0.8848 (0.1323) & 0.8570 (0.1513) & 0.9175 (0.0925)
			& 0.8631 (0.1513) & 0.9032 (0.1356) & \textbf{0.9447 (0.0854)} \\
			& ARI
			& 0.7729 (0.2010) & 0.7177 (0.2406) & 0.8062 (0.1659)
			& 0.7324 (0.2127) & 0.8242 (0.2112) & \textbf{0.8695 (0.1577)} \\
			\bottomrule
		\end{tabular}%
	}
\end{table}

The full ten-class problem is considerably more difficult than the reduced two- and three-class problems.  In the full-data analysis, $K$-median attains the largest accuracy and ARI, while GEM is a close second with accuracy $0.7512$ and ARI $0.6304$. The comparison changes in the reduced-subset analyses. Over the $45$ pairwise problems, GEM attains the highest average accuracy and ARI, with values $0.9785$ and $0.9239$, respectively. Over the $120$ triplet problems, GEM again ranks first, with average accuracy $0.9447$ and average ARI $0.8695$.  These results suggest that the semiparametric elliptical update is particularly effective once the clustering task is not dominated by simultaneous separation of many visually similar digit classes.

We also examine whether a Gaussian working model is adequate within each digit class.  For digit class $k\in\{0,\ldots,9\}$ and coordinate $j$, let
\[
  z_{ij}^{(k)}
  =
  \frac{x_{ij}-\bar x_{k,j}}{s_{k,j}},
\]
where $\bar x_{k,j}$ and $s_{k,j}$ are the within-class sample mean and sample standard deviation of feature $j$.  For each class $k$, we pool the standardized entries $z_{ij}^{(k)}$ across observations and coordinates, and compare their empirical quantiles with standard normal quantiles.  This class-wise diagnostic is more informative than a single pooled diagnostic, because the full data set is a ten-component mixture and is not expected to resemble a single Gaussian population.

\begin{figure}[H]
\centering
\includegraphics[width=0.98\textwidth]{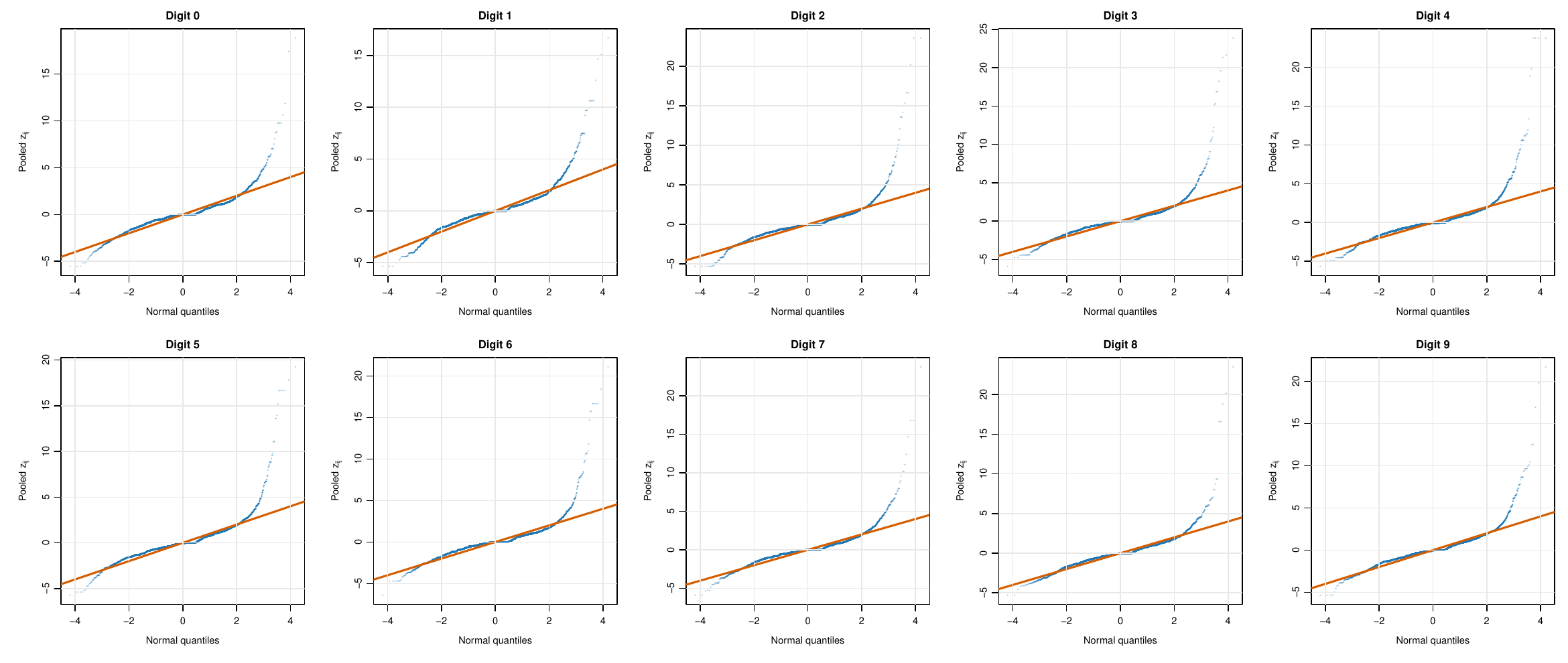}
\caption{Class-wise pooled QQ plots for the Optdigits data.  For each digit class $g$, the standardized entries $z_{ij}^{(k)}=(x_{ij}-\bar x_{k,j})/s_{k,j}$ are pooled across observations and coordinates and compared with standard normal quantiles.}
\label{fig:optdigits-classwise-zij-qq}
\end{figure}

Figure~\ref{fig:optdigits-classwise-zij-qq} shows systematic but moderate right-tail deviations from Gaussianity in most classes.  The ratio between the empirical $0.99$ quantile of $|z_{ij}^{(k)}|$ and its Gaussian counterpart ranges from $1.07$ to $1.23$, with median approximately $1.15$ across the ten classes.  Thus the class-wise marginal distributions are not extremely heavy tailed, but they are also not well described by a purely Gaussian working model.  This pattern is consistent with the empirical comparison in Table~\ref{tab:optdigits}: robust location procedures are competitive, and GEM gives the strongest average performance in the pairwise and triplet analyses while remaining close to the best method on the full ten-class problem.

\section{Conclusion}
\label{sec:discussion}
This paper develops a semiparametric GEM procedure for high-dimensional clustering under heavy-tailed elliptical mixture models.  The method combines a nonparametric update of the common radial generator with robust TME--POET--GLASSO estimation of a common sparse precision-shape matrix, and it is equipped with high-dimensional consistency theory for parameter estimation and plug-in clustering.  Simulations and the Optdigits analysis show that the proposed procedure is especially effective when the data are non-Gaussian but retain approximate elliptical geometry.

Several extensions are natural.  The most direct one is to relax the common-shape assumption and allow component-specific shape matrices $\mSigma_k$ or sparse precision matrices $\mP_k$.  A statistically stable version would likely require a joint graphical penalty, shared support restrictions, or low-rank-plus-sparse coupling across clusters, because unrestricted component-specific shapes are difficult to identify and estimate when $p$ is large.  A second direction is to allow component-specific generators $g_k$, which would accommodate clusters with different tail behavior; however, this extension requires additional identifiability conditions separating radial variation from shape variation, since changes in $g_k$ and rescalings of $\mSigma_k$ can otherwise produce observationally similar densities.  A third direction is to develop a full theory for data-driven cluster-number selection under the proposed gap-LSE rule and to compare it with semiparametric information criteria.  Related tuning questions include adaptive bandwidth selection, factor-rank selection, and graph-penalty selection under heavy-tailed elliptical sampling.  Finally, minimax lower bounds for high-dimensional semiparametric elliptical mixtures remain open and would clarify whether the rates obtained here are optimal.

\appendix
\section{Proofs of the main results}
\label{sec:proofs}
\subsection{Preliminary lemmas}

\begin{lemma}[Bernstein inequality for bounded empirical classes]
\label{lem:bernstein}
Let $\cF$ be a VC-type class with entropy bound \eqref{eq:entropy_assump}, envelope $F$, and variance bound $\sup_{f\in\cF}\Var\{f(\bm X)\}\le \sigma^2$. If $\|F\|_\infty\le B$, then there exist constants $C,c>0$ such that for every $t>0$,
\begin{equation}
\label{eq:bernstein_lemma}
 \Prob\Biggl[
 \sup_{f\in\cF}
 \abs{\frac1n\sum_{i=1}^n\{f(\bm X_i)-\E f(\bm X)\}}
 > C\Biggl\{\sigma\sqrt{\frac{\log p+t}{n}} + B\frac{\log p+t}{n}\Biggr\}
 \Biggr]
 \le e^{-ct}.
\end{equation}
\end{lemma}
The inequality is a standard consequence of Bousquet's version of Talagrand's inequality for suprema of bounded empirical processes combined with VC-type entropy bounds; see \citet{Bousquet2002Bennett} and \citet[Sections~2.2 and~2.6]{VanderVaartWellner1996}.

\begin{lemma}[Elliptical geometry inside the local ball]
\label{lem:elliptic_geometry}
Suppose Assumptions~\ref{ass:model}--\ref{ass:local} hold. Then there exists $C<\infty$ such that for every $z\in\cB(r_0)$ and every $1\le k\le K$,
\begin{equation}
\label{eq:delta_vs_R}
 C^{-1}(1+R^2)
 \le
 1+\delta_k(\bm X;z)
 \le
 C(1+R^2)
 \qquad\text{almost surely,}
\end{equation}
whenever $Z=k$. Consequently,
\begin{equation}
\label{eq:score_times_residual}
 w_z\bigl(\delta_k(\bm X;z)\bigr)\,\norminf{\bm X-\bmu_k}
 \le C,
\end{equation}
and for every $\mSigma\in\cC_\Sigma$,
\begin{equation}
\label{eq:normalized_outer_bound}
 \normmax{\frac{(\bm X-\bmu_k)(\bm X-\bmu_k)^\top}{(\bm X-\bmu_k)^\top \mSigma^{-1}(\bm X-\bmu_k)}}
 \le C.
\end{equation}
\end{lemma}

\begin{proof}
Fix $z\in\cB(r_0)$ and condition on $Z=k$. Write
\[
 \bm X-\bmu_k
 =
 \mSigma^{*1/2}R\bm U + (\bmu_k^*-\bmu_k).
\]
Since $\lambda_{\min}(\mP)\ge m_\Omega/2$ and $\lambda_{\max}(\mP)\le 2M_\Omega$ on $\cB(r_0)$,
\[
 \frac{m_\Omega}{2}\normtwo{\bm X-\bmu_k}^2
 \le
 \delta_k(\bm X;z)
 \le
 2M_\Omega\normtwo{\bm X-\bmu_k}^2.
\]
Also,
\[
 \normtwo{\bm X-\bmu_k}
 \le
 \normtwo{\mSigma^{*1/2}R\bm U} + \normtwo{\bmu_k^*-\bmu_k}
 \le
 \lambda_{\max}(\mSigma^*)^{1/2}R + p^{1/2}r_0,
\]
and similarly
\[
 \normtwo{\bm X-\bmu_k}
 \ge
 \lambda_{\min}(\mSigma^*)^{1/2}R - p^{1/2}r_0.
\]
Because $r_0$ is fixed and the eigenvalues of $\mSigma^*$ are bounded away from $0$ and $\infty$, the last two displays imply \eqref{eq:delta_vs_R}. Next,
\[
 w_z\bigl(\delta_k(\bm X;z)\bigr)\,\norminf{\bm X-\bmu_k}
 \le
 w_z\bigl(\delta_k(\bm X;z)\bigr)\,\normtwo{\bm X-\bmu_k}
 \le
 M_w\frac{\normtwo{\bm X-\bmu_k}}{\sqrt{1+\delta_k(\bm X;z)}}
 \le C,
\]
which proves \eqref{eq:score_times_residual}. Finally, for $\mSigma\in\cC_\Sigma$,
\[
 (\bm X-\bmu_k)^\top\mSigma^{-1}(\bm X-\bmu_k)
 \ge \lambda_{\min}(\mSigma^{-1})\normtwo{\bm X-\bmu_k}^2
 = \frac{1}{\lambda_{\max}(\mSigma)}\normtwo{\bm X-\bmu_k}^2,
\]
so
\[
 \abs{\frac{(X_a-\mu_{k,a})(X_b-\mu_{k,b})}{(\bm X-\bmu_k)^\top\mSigma^{-1}(\bm X-\bmu_k)}}
 \le
 \lambda_{\max}(\mSigma)
 \frac{\abs{(X_a-\mu_{k,a})(X_b-\mu_{k,b})}}{\normtwo{\bm X-\bmu_k}^2}
 \le \lambda_{\max}(\mSigma)
 \le C.
\]
Taking the maximum over $a$ and $b$ gives \eqref{eq:normalized_outer_bound}.
\end{proof}

\begin{lemma}[Softmax derivative]
\label{lem:softmax}
For $\bm a=(a_1,\dots,a_K)^\top$, define
\[
 s_k(\bm a)=\frac{\exp(a_k)}{\sum_{\ell=1}^K\exp(a_\ell)}.
\]
Then
\begin{equation}
\label{eq:softmax_derivative}
 \frac{\partial s_k(\bm a)}{\partial a_\ell}
 = s_k(\bm a)\{\bone(k=\ell)-s_\ell(\bm a)\}.
\end{equation}
Consequently,
\begin{equation}
\label{eq:softmax_l1}
 \sum_{\ell=1}^K \abs{\frac{\partial s_k(\bm a)}{\partial a_\ell}}
 \le 2 s_k(\bm a)\{1-s_k(\bm a)\}
 \le 2\max_{1\le u\le K}s_u(\bm a)\{1-s_u(\bm a)\}.
\end{equation}
\end{lemma}

\begin{proof}
Differentiate the quotient definition of $s_k$:
\[
 \frac{\partial s_k(\bm a)}{\partial a_\ell}
 =
 \frac{\bone(k=\ell)e^{a_k}\sum_{u=1}^K e^{a_u} - e^{a_k}e^{a_\ell}}{(\sum_{u=1}^K e^{a_u})^2}
 = s_k(\bm a)\{\bone(k=\ell)-s_\ell(\bm a)\}.
\]
Summing over $\ell$ gives
\[
 \sum_{\ell=1}^K \abs{\frac{\partial s_k(\bm a)}{\partial a_\ell}}
 = s_k(\bm a)\Bigl\{1-s_k(\bm a)+\sum_{\ell\neq k}s_\ell(\bm a)\Bigr\}
 = 2s_k(\bm a)\{1-s_k(\bm a)\},
\]
which is \eqref{eq:softmax_l1}.
\end{proof}

\begin{lemma}[POET perturbation]
\label{lem:poet_perturb}
Suppose Assumption~\ref{ass:factor} holds. There exist $\epsilon_{\mathrm P}>0$ and $C_{\mathrm{P}}<\infty$, independent of $(n,p)$, such that whenever the symmetric matrices $\mA$ and $\mB$ satisfy
\[
 \normmax{\mA-\mSigma^*}\vee \normmax{\mB-\mSigma^*}\le \epsilon_{\mathrm P},
\]
one has
\begin{equation}
\label{eq:poet_perturb}
 \normmax{\operatorname{POET}_{\lambda_u,m_*}(\mA)-\operatorname{POET}_{\lambda_u,m_*}(\mB)}
 \le C_{\mathrm{P}}\normmax{\mA-\mB} + C_{\mathrm{P}}\rho_{n,p}.
\end{equation}
\end{lemma}

\begin{proof}
This is precisely the local entrywise stability condition \eqref{eq:poet_stability_assump} in Assumption~\ref{ass:factor}.  The statement is recorded as a lemma only to keep the proof of the TME--POET--GLASSO theorem modular.  Importantly, no step of the argument bounds the spectral norm by $p\normmax{\mA-\mB}$; hence the constant $C_{\mathrm P}$ in \eqref{eq:poet_perturb} is dimension-free under the stated POET stability assumption.
\end{proof}

\begin{lemma}[GLASSO perturbation]
\label{lem:glasso_perturb}
Suppose Assumption~\ref{ass:glasso} holds. Then there exist constants $\epsilon_G>0$ and $C_G<\infty$ such that whenever
\[
 \normmax{\widehat{\mSigma}-\mSigma}\le \epsilon_G,
\]
and both matrices are positive definite, the corresponding graphical-lasso solutions satisfy
\begin{equation}
\label{eq:glasso_perturb_lemma}
 \normone{\widehat{\mP}-\mP}
 \le C_G s_\Omega\Bigl(\normmax{\widehat{\mSigma}-\mSigma}+\lambda_{\Omega,n}\Bigr).
\end{equation}
\end{lemma}

\begin{proof}
The Karush--Kuhn--Tucker system for the graphical-lasso problem is
\[
 \widehat{\mSigma}-\widehat{\mP}^{-1}+\lambda_{\Omega,n}\widehat{\mathbf Z}=\bzero,
 \qquad
 \mSigma-\mP^{-1}+\lambda_{\Omega,n}\mathbf Z=\bzero,
\]
where $\widehat{\mathbf Z}$ and $\mathbf Z$ are subgradient matrices with $\widehat Z_{ab}=\operatorname{sgn}(\widehat P_{ab})$ and $Z_{ab}=\operatorname{sgn}(P_{ab})$ on the support, and entries in $[-1,1]$ off the support. Subtracting the two equations gives
\[
 \widehat{\mSigma}-\mSigma
 =
 \widehat{\mP}^{-1}-\mP^{-1}-\lambda_{\Omega,n}(\widehat{\mathbf Z}-\mathbf Z).
\]
Vectorizing and expanding $\widehat{\mP}^{-1}-\mP^{-1}$ around $\mP$ yields
\[
 \operatorname{vec}(\widehat{\mSigma}-\mSigma)
 =
 -\mathbf\Gamma^*\operatorname{vec}(\widehat{\mP}-\mP) + \operatorname{vec}(\mathbf R) - \lambda_{\Omega,n}\operatorname{vec}(\widehat{\mathbf Z}-\mathbf Z),
\]
with a quadratic remainder $\mathbf R$. By the incoherence condition \eqref{eq:irrep}, the standard primal--dual witness argument yields
\[
 \norminf{\operatorname{vec}(\widehat{\mP}-\mP)_{S_\Omega}}
 \le
 M_\Gamma\Bigl\{\normmax{\widehat{\mSigma}-\mSigma}+\lambda_{\Omega,n}\Bigr\},
\]
and the off-support coordinates are dominated by the same quantity. Converting the vector bound to the matrix $\ell_1$ norm and using the row sparsity $s_\Omega$ gives \eqref{eq:glasso_perturb_lemma}.
\end{proof}

\subsection{Proof of Proposition~\ref{prop:local_regularity} and Theorem~\ref{thm:population}}

\begin{proof}[Proof of Proposition~\ref{prop:local_regularity}]
Fix $z,\widetilde z\in\cB(r_1)$ and abbreviate
\[
 \eta_k(\bx;z)=\log \pi_k + \ell\bigl(\delta_k(\bx;z)\bigr),
 \qquad
 \widetilde\eta_k(\bx)=\eta_k(\bx;\widetilde z).
\]
By Lemma~\ref{lem:softmax},
\begin{align}
 \abs{r_k(\bx;z)-r_k(\bx;\widetilde z)}
 &\le \int_0^1 \sum_{\ell=1}^K \abs{\frac{\partial s_k}{\partial a_\ell}\bigl(\bm a_t(\bx)\bigr)}\,dt\; \max_{1\le \ell\le K}\abs{\eta_\ell(\bx;z)-\eta_\ell(\bx;\widetilde z)} \notag\\
 &\le 2\sup_{t\in[0,1]}\Gamma\bigl(\bx;z_t\bigr)\max_{1\le \ell\le K}\abs{\eta_\ell(\bx;z)-\eta_\ell(\bx;\widetilde z)},
 \label{eq:r_diff_1}
\end{align}
where $z_t=t z +(1-t)\widetilde z$ and
\[
 \Gamma(\bx;z)=\max_{1\le k\le K}r_k(\bx;z)\{1-r_k(\bx;z)\}.
\]
Next,
\begin{align}
 \abs{\eta_k(\bx;z)-\eta_k(\bx;\widetilde z)}
 &\le c_\pi^{-1}\norminf{\bpi-\widetilde\bpi}
 + \abs{\ell\{\delta_k(\bx;z)\}-\ell\{\delta_k(\bx;\widetilde z)\}} 
\notag\\
 &\qquad + \sup_{0\le u\le U_0}\abs{\ell(u)-\widetilde\ell(u)}.
 \label{eq:eta_diff_1}
\end{align}
For the quadratic part,
\begin{align}
 \delta_k(\bx;z)-\delta_k(\bx;\widetilde z)
 &= (\bx-\bmu_k)^\top(\mP-\widetilde\mP)(\bx-\bmu_k)
 + (\bx-\bmu_k)^\top\widetilde\mP(\widetilde\bmu_k-\bmu_k) \notag\\
 &\qquad + (\widetilde\bmu_k-\bmu_k)^\top\widetilde\mP(\bx-\widetilde\bmu_k).
 \label{eq:delta_expand}
\end{align}
By Assumptions~\ref{ass:model} and \ref{ass:local},
\[
 \abs{\delta_k(\bx;z)-\delta_k(\bx;\widetilde z)}
 \le C\{1+\norminf{\bx}\}^2 d(z,\widetilde z).
\]
Combining this with the score decay in Assumption~\ref{ass:generator} and the elliptical bound \eqref{eq:delta_vs_R} gives
\begin{equation}
\label{eq:eta_diff_2}
 \abs{\eta_k(\bx;z)-\eta_k(\bx;\widetilde z)}
 \le C(1+R) d(z,\widetilde z).
\end{equation}
Insert \eqref{eq:eta_diff_2} into \eqref{eq:r_diff_1}:
\begin{equation}
\label{eq:r_diff_2}
 \abs{r_k(\bx;z)-r_k(\bx;\widetilde z)}
 \le C\Gamma(\bx;z_t)(1+R)d(z,\widetilde z).
\end{equation}
Under Assumption~\ref{ass:model}, pointwise posterior overlap goes to zero as $\Delta_0\to\infty$. Since $0\le \Gamma\le 1/4$ and Assumption~\ref{ass:model} now imposes $\E(1+R)<\infty$, dominated convergence with envelope $(1+R)/4$ shows
\[
 \bar\Gamma_1(\Delta_0)
 :=
 \E\Bigl[(1+R)\sup_{z\in\cB(r_1)}\Gamma(\bm X;z)\Bigr]
 \longrightarrow 0
 \qquad\text{as }\Delta_0\to\infty.
\]
Therefore
\begin{equation}
\label{eq:Pi_Lip}
 \norminf{\bPi(z)-\bPi(\widetilde z)}
 \le C\bar\Gamma_1(\Delta_0)d(z,\widetilde z).
\end{equation}
The generator block is controlled through the transformed-radius density. Let $f_{Y,z}$ denote the density of $q\{\delta_k(\bm X;z)\}$ under the responsibility weights. By the change of variables in Section~\ref{sec:model}, the map $f\mapsto \cA^{-1}(f)$ is linear and bounded on $L_\infty[0,\log(1+U_0)]$. Hence
\begin{equation}
\label{eq:G_Lip_1}
 \sup_{0\le u\le U_0}\abs{G(z)(u)-G(\widetilde z)(u)}
 \le C\sup_{0\le y\le \log(1+U_0)}\abs{f_{Y,z}(y)-f_{Y,\widetilde z}(y)}.
\end{equation}
The density difference is bounded by the same posterior-weight difference and the same quadratic-form perturbation as in \eqref{eq:r_diff_2}, so
\begin{equation}
\label{eq:G_Lip_2}
 \sup_{0\le u\le U_0}\abs{G(z)(u)-G(\widetilde z)(u)}
 +
 \sup_{0\le u\le U_0}(1+u)^{1/2}\abs{w_z(u)-w_{\widetilde z}(u)}
 \le C\bar\Gamma_1(\Delta_0)d(z,\widetilde z).
\end{equation}
This proves \eqref{eq:local_regularity}.

For the center block, write
\[
 N_k(z)=\E\bigl[r_k(\bm X;z)w_z\{\delta_k(\bm X;z)\}\bm X\bigr],
 \qquad
 D_k(z)=\E\bigl[r_k(\bm X;z)w_z\{\delta_k(\bm X;z)\}\bigr].
\]
By \eqref{eq:score_lower}, \eqref{eq:delta_vs_R}, and $\Prob(R^2\le u_0)\ge c_0$,
\[
 \inf_{z\in\cB(r_1)} D_k(z)
 \ge c_\pi c_0 m_w/4
 =: d_0>0.
\]
Using \eqref{eq:score_times_residual}, \eqref{eq:r_diff_2}, and \eqref{eq:G_Lip_2},
\[
 \norminf{N_k(z)-N_k(\widetilde z)} + \abs{D_k(z)-D_k(\widetilde z)}
 \le C d(z,\widetilde z).
\]
Since $\widetilde U_k(z)=N_k(z)/D_k(z)$,
\begin{align*}
 \norminf{\widetilde U_k(z)-\widetilde U_k(\widetilde z)}
 &\le \frac{\norminf{N_k(z)-N_k(\widetilde z)}}{D_k(z)}
 + \norminf{N_k(\widetilde z)}\abs{\frac{1}{D_k(z)}-\frac{1}{D_k(\widetilde z)}} \\
 &\le C d(z,\widetilde z).
\end{align*}
Damping by \eqref{eq:oracle_pi_mu} gives \eqref{eq:center_regularity}.

For the Tyler block define $\Psi$ by \eqref{eq:Psi_def}. Assumption~\ref{ass:tyler} implies that $D_{\mSigma}\operatorname{vec}\Psi(\mSigma^*,z^*)$ is invertible. Hence the implicit function theorem yields a neighborhood of $(\mSigma^*,z^*)$ on which the equation $\Psi(\mSigma,z)=\bzero$ defines a continuously differentiable map $z\mapsto \mSigma_{\mathrm{Ty}}(z)$. Its derivative is
\[
 D_z\operatorname{vec}\mSigma_{\mathrm{Ty}}(z)
 =
 -\Bigl[D_{\mSigma}\operatorname{vec}\Psi\{\mSigma_{\mathrm{Ty}}(z),z\}\Bigr]^{-1}
 D_z\operatorname{vec}\Psi\{\mSigma_{\mathrm{Ty}}(z),z\}.
\]
The inverse Jacobian is uniformly bounded on a sufficiently small neighborhood, and the derivative with respect to $z$ is bounded because of \eqref{eq:normalized_outer_bound}. Therefore
\[
 \normmax{\mSigma_{\mathrm{Ty}}(z)-\mSigma_{\mathrm{Ty}}(\widetilde z)}
 \le L_{\mathrm{Ty}} d(z,\widetilde z).
\]
Apply Lemma~\ref{lem:poet_perturb} to the second POET step, and then Lemma~\ref{lem:glasso_perturb} to the GLASSO map, to obtain the second bound in \eqref{eq:precision_regularity}. The proposition is proved.
\end{proof}

\begin{proof}[Proof of Theorem~\ref{thm:population}]
By construction of the oracle map, the target parameter is a fixed point, so $M(z^*)=z^*$.  Hence for $z\in\cB(r)$,
\[
 d\bigl(M(z),z^*\bigr)
 =
 d\bigl(M(z),M(z^*)\bigr).
\]
Apply Proposition~\ref{prop:local_regularity} with $\widetilde z=z^*$:
\[
 \norminf{\bPi(z)-\bpi^*}
 \le L_\pi^\dagger d(z,z^*),
\]
\[
 \sup_{0\le u\le U_0}\abs{G(z)(u)-g^*(u)}
 \vee
 \sup_{0\le u\le U_0}(1+u)^{1/2}\abs{w_z(u)-w^*(u)}
 \le L_G^\dagger d(z,z^*),
\]
\[
 \max_{k\le K}\norminf{U_k(z)-\bmu_k^*}
 \le (1-\eta_\mu+\eta_\mu L_\mu^\dagger)d(z,z^*),
\]
\[
 \normone{T(z)-\mP^*}
 \le \eta_\Omega L_\Omega d(z,z^*).
\]
Taking the maximum of the four bounds gives
\[
 d\bigl(M(z),z^*\bigr)
 \le
 \Bigl[\max\{L_\pi^\dagger,L_G^\dagger,1-\eta_\mu+\eta_\mu L_\mu^\dagger,\eta_\Omega L_\Omega\}+C_r r\Bigr]d(z,z^*).
\]
By the final sentence of Proposition~\ref{prop:local_regularity}, the constants $L_\pi^\dagger$, $L_G^\dagger$, and $L_\mu^\dagger$ are controlled by the overlap coefficient and therefore tend to $0$ as the separation level $\Delta_0$ grows. Choosing $\Delta_0\ge \Delta_0^*$ and then $r\le r_*$ so that the bracketed constant is strictly smaller than one yields \eqref{eq:population_contraction}. The explicit expression \eqref{eq:kappa_formula} is the same bracketed constant.
\end{proof}

\subsection{Proof of Proposition~\ref{prop:bias} and Theorem~\ref{thm:sample_nonprecision}}

\begin{proof}[Proof of Proposition~\ref{prop:bias}]
Let $f_{Y,z}$ be the transformed-radius density under $z$, and let
\[
 f_{Y,z,h}(y)=\int K_h(y-v)f_{Y,z}(v)\,dv,
 \qquad K_h(t)=h^{-1}K(t/h).
\]
Because $K$ is symmetric and of order two,
\[
 \int K(u)\,du=1,
 \qquad \int uK(u)\,du=0,
 \qquad \mu_2(K)=\int u^2K(u)\,du<\infty.
\]
For each fixed $y$,
\begin{align*}
 f_{Y,z,h}(y)-f_{Y,z}(y)
 &= \int K(u)\{f_{Y,z}(y-hu)-f_{Y,z}(y)\}\,du \\
 &= \frac{h^2}{2}\int u^2K(u)f_{Y,z}^{\prime\prime}(y-\theta_{y,u}hu)\,du
\end{align*}
for some measurable $\theta_{y,u}\in(0,1)$. Hence,
\[
 \abs{f_{Y,z,h}(y)-f_{Y,z}(y)}
 \le \frac{h^2}{2}\mu_2(K)\sup_{v}\abs{f_{Y,z}^{\prime\prime}(v)}.
\]
By Assumption~\ref{ass:generator},
\[
 \sup_{z\in\cB(r)}\sup_y \abs{f_{Y,z,h}(y)-f_{Y,z}(y)}
 \le C h^2.
\]
The reconstruction map from the transformed-density scale to the generator scale is linear, and the smoothing-spline operator is uniformly bounded on the deterministic grid sequence specified in Assumption~\ref{ass:entropy}. Hence
\[
 \sup_{z\in\cB(r)}\sup_{0\le u\le U_0}\abs{G_h(z)(u)-G(z)(u)}
 +
 \sup_{z\in\cB(r)}\sup_{0\le u\le U_0}(1+u)^{1/2}\abs{w_h(z)-w_z(u)}
 \le C h^2.
\]
Insert these bounds into the formulas for $\Pi$, $U_g$, and $T$. Only the generator and center blocks depend on $h$, so the same order $h^2$ propagates to the whole map. This proves \eqref{eq:bias_bound}.
\end{proof}

\begin{proof}[Proof of Theorem~\ref{thm:sample_nonprecision}]
For the mixing-proportion block define
\[
 \widehat \Pi_k(z)=\frac1n\sum_{i=1}^n r_k(\bm X_i;z),
 \qquad
 \Pi_k(z)=\E\{r_k(\bm X;z)\}.
\]
The envelope of $\cF_{\pi,k}$ is $1$. Applying Lemma~\ref{lem:bernstein} with $t=c\log p$ gives
\begin{equation}
\label{eq:mixing_empirical}
 \Prob\left[
 \sup_{z\in\cB(r)}\norminf{\widehat{\bpi}(z)-\bPi(z)}
 > C\sqrt{\frac{\log p}{n}}
 \right]
 \le Cp^{-c}.
\end{equation}

For the generator block, fix the deterministic transformed-radius grid $y_1,\dots,y_{M_n}$ used in the theory. For a single grid point $y_m$,
\[
 \widehat f_{Y,z,h}(y_m)
 =
 \frac{1}{nh}\sum_{i=1}^n\sum_{k=1}^K r_k(\bm X_i;z)K\Bigl(\frac{y_m-q\{\delta_k(\bm X_i;z)\}}{h}\Bigr).
\]
The envelope of $\cF_{K,m,h}$ is of order $h^{-1}$ and its variance is of order $h^{-1}$. Lemma~\ref{lem:bernstein} therefore yields
\[
 \Prob\left[
 \sup_{z\in\cB(r)}\abs{\widehat f_{Y,z,h}(y_m)-f_{Y,z,h}(y_m)}
 > C\Biggl\{\sqrt{\frac{\log n}{nh}}+\frac{\log n}{nh}\Biggr\}
 \right]
 \le Cn^{-c}.
\]
Because $nh/\log n\to\infty$ by Assumption~\ref{ass:init}, the second term is absorbed by the first, and therefore
\[
 \Prob\left[
 \sup_{z\in\cB(r)}\abs{\widehat f_{Y,z,h}(y_m)-f_{Y,z,h}(y_m)}
 > C\sqrt{\frac{\log n}{nh}}
 \right]
 \le Cn^{-c}.
\]
Since $\log M_n\lesssim\log n$, a union bound over $1\le m\le M_n$ only changes the constant in the polynomial probability bound and gives
\begin{equation}
\label{eq:kde_grid_bound}
\begin{aligned}
 &\Prob\Biggl[
 \sup_{z\in\cB(r)}\max_{1\le m\le M_n}\abs{\widehat f_{Y,z,h}(y_m)-f_{Y,z,h}(y_m)}
 > C\sqrt{\frac{\log n}{nh}}
 \Biggr]
 \\
 &\qquad\le Cn^{-c}.
\end{aligned}
\end{equation}
The smoothing spline and the linear interpolation step are uniformly bounded on the same deterministic grid sequence. Hence \eqref{eq:kde_grid_bound} implies
\begin{equation}
\label{eq:generator_empirical}
\begin{aligned}
 &\Prob\Biggl[
 \sup_{z\in\cB(r)}
 \Biggl\{
 \sup_{0\le u\le U_0}\abs{\widehat G_{n,h}(z)(u)-G_h(z)(u)}
 \\
 &\qquad\vee
 \sup_{0\le u\le U_0}(1+u)^{1/2}\abs{\widehat w_{n,h}(z)(u)-w_h(z)(u)}
 \Biggr\}
 > C\sqrt{\frac{\log n}{nh}}
 \Biggr]
 \\
 &\qquad\le Cn^{-c}.
\end{aligned}
\end{equation}

For the center block write
\[
 \widehat N_{k,j}(z)
 =
 \frac1n\sum_{i=1}^n r_k(\bm X_i;z)\widehat w_{n,h}\bigl(\delta_k(\bm X_i;z)\bigr)X_{ij},
 \qquad
 N_{k,j,h}(z)=\E\bigl[r_k(\bm X;z)w_h\bigl(\delta_k(\bm X;z)\bigr)X_j\bigr],
\]
and
\[
 \widehat D_k(z)
 =
 \frac1n\sum_{i=1}^n r_k(\bm X_i;z)\widehat w_{n,h}\bigl(\delta_k(\bm X_i;z)\bigr),
 \qquad
 D_{k,h}(z)=\E\bigl[r_k(\bm X;z)w_h\bigl(\delta_k(\bm X;z)\bigr)\bigr].
\]
By Lemma~\ref{lem:elliptic_geometry}, the envelope of the class $\cF_{\mu,k,j,h}$ is bounded by a constant that does not depend on $(n,p)$. Lemma~\ref{lem:bernstein} therefore yields
\begin{equation}
\label{eq:center_numerators}
 \Prob\left[
 \sup_{z\in\cB(r)}\max_{k\le K,\,j\le p}\abs{\widehat N_{k,j}(z)-N_{k,j,h}(z)}
 > C\sqrt{\frac{\log p}{n}}
 \right]
 \le Cp^{-c}.
\end{equation}
The same argument gives
\begin{equation}
\label{eq:center_denominators}
 \Prob\left[
 \sup_{z\in\cB(r)}\max_{k\le K}\abs{\widehat D_k(z)-D_{k,h}(z)}
 > C\sqrt{\frac{\log p}{n}}
 \right]
 \le Cp^{-c}.
\end{equation}
Because of \eqref{eq:score_lower}, the population denominators are bounded away from zero. Intersecting the events in \eqref{eq:generator_empirical}, \eqref{eq:center_numerators}, and \eqref{eq:center_denominators} gives
\[
 \max_{k\le K}\norminf{\widehat{\bU}_k(z)-\bU_{k,h}(z)}
 \le C\Biggl\{\sqrt{\frac{\log p}{n}}+\sqrt{\frac{\log n}{nh}}\Biggr\}
\]
uniformly over $z\in\cB(r)$. Combining this with \eqref{eq:mixing_empirical} and \eqref{eq:generator_empirical} proves \eqref{eq:sample_nonprecision}.
\end{proof}

\subsection{Proof of Theorem~\ref{thm:tme_glasso}}

\begin{proof}[Proof of Theorem~\ref{thm:tme_glasso}]
The proof begins with the weighted spatial-sign pilot, but the final conclusion is about the weighted Tyler estimator. Define
\[
 \widehat{\mH}_n(z)
 =
 \frac1n\sum_{i=1}^n\sum_{k=1}^K
 r_k(\bm X_i;z)
 \frac{\br_k^+(\bm X_i;z)\br_k^+(\bm X_i;z)^\top}{\normtwo{\br_k^+(\bm X_i;z)}^2\vee \varepsilon_r},
\]
and let $\mH(z)=\E\widehat{\mH}_n(z)$. By \eqref{eq:normalized_outer_bound}, every entry of the summand is bounded. Applying Lemma~\ref{lem:bernstein} to the class generated by the matrix entries gives
\begin{equation}
\label{eq:H_concentration}
 \Prob\left[
 \sup_{z\in\cB(r)}\normmax{\widehat{\mH}_n(z)-\mH(z)}
 > C\sqrt{\frac{\log p}{n}}
 \right]
 \le Cp^{-c}.
\end{equation}
Apply Lemma~\ref{lem:poet_perturb} to the first POET step:
\begin{equation}
\label{eq:pss_concentration}
 \Prob\left[
 \sup_{z\in\cB(r)}\normmax{\widehat{\mSigma}_{\mathrm{pss}}(z)-\mSigma_{\mathrm{pss}}(z)}
 > C\Biggl\{\sqrt{\frac{\log p}{n}}+\rho_{n,p}\Biggr\}
 \right]
 \le Cp^{-c}.
\end{equation}
This event places the sample Tyler iteration in the same local neighborhood as the population Tyler solution.

Define the sample Tyler score
\begin{equation}
\label{eq:Psi_n}
 \Psi_n(\mSigma,z)
 =
 \mSigma
 -
 \Normalize\Biggl[
 (1-\rho_T)
 \frac{p}{n}\sum_{i=1}^n\sum_{k=1}^K
 r_k(\bm X_i;z)
 \frac{\br_k^+(\bm X_i;z)\br_k^+(\bm X_i;z)^\top}{\br_k^+(\bm X_i;z)^\top \mSigma^{-1}\br_k^+(\bm X_i;z)\vee \varepsilon_r}
 +\rho_T\mI
 \Biggr].
\end{equation}
The sample Tyler estimator is the local solution of $\Psi_n(\mSigma,z)=\bzero$, and the population Tyler estimator is the local solution of $\Psi(\mSigma,z)=\bzero$. Therefore
\[
 \bzero
 =
 \Psi_n\{\widehat{\mSigma}_{\mathrm{Ty}}(z),z\}
 -
 \Psi\{\mSigma_{\mathrm{Ty}}(z),z\}.
\]
Add and subtract $\Psi_n\{\mSigma_{\mathrm{Ty}}(z),z\}$:
\begin{equation}
\label{eq:Ty_expand_1}
 \Psi_n\{\widehat{\mSigma}_{\mathrm{Ty}}(z),z\}
 -
 \Psi_n\{\mSigma_{\mathrm{Ty}}(z),z\}
 =
 -\Bigl[\Psi_n\{\mSigma_{\mathrm{Ty}}(z),z\}-\Psi\{\mSigma_{\mathrm{Ty}}(z),z\}\Bigr].
\end{equation}
Vectorize and apply the mean-value theorem. There exists a matrix $\overline{\mSigma}(z)$ on the line segment joining $\widehat{\mSigma}_{\mathrm{Ty}}(z)$ and $\mSigma_{\mathrm{Ty}}(z)$ such that
\begin{equation}
\label{eq:Ty_expand_2}
\begin{aligned}
 \operatorname{vec}\Bigl\{\widehat{\mSigma}_{\mathrm{Ty}}(z)-\mSigma_{\mathrm{Ty}}(z)\Bigr\}
 & =
 -\Bigl[D_{\mSigma}\operatorname{vec}\Psi_n\{\overline{\mSigma}(z),z\}\Bigr]^{-1} \\
 &\quad \times
 \operatorname{vec}\Bigl[\Psi_n\{\mSigma_{\mathrm{Ty}}(z),z\}-\Psi\{\mSigma_{\mathrm{Ty}}(z),z\}\Bigr].
\end{aligned}
\end{equation}
Assumption~\ref{ass:tyler} and the pilot concentration \eqref{eq:pss_concentration} imply that the inverse Jacobian on the right-hand side is uniformly bounded. It remains to bound the empirical score term. For a fixed entry $(a,b)$,
\[
 \xi_{ik}^{ab}(z,\mSigma)
 =
 r_k(\bm X_i;z)
 \frac{r_{k,a}^+(\bm X_i;z)r_{k,b}^+(\bm X_i;z)}{\br_k^+(\bm X_i;z)^\top \mSigma^{-1}\br_k^+(\bm X_i;z)\vee \varepsilon_r}
\]
is uniformly bounded by Lemma~\ref{lem:elliptic_geometry}. Applying Lemma~\ref{lem:bernstein} to the class $\cF_{\mathrm{Ty},a,b}$ gives
\[
 \Prob\left[
 \sup_{z\in\cB(r)}\sup_{\mSigma\in\cC_\Sigma}
 \abs{\frac1n\sum_{i=1}^n\sum_{k=1}^K \xi_{ik}^{ab}(z,\mSigma)-\E\xi_{ik}^{ab}(z,\mSigma)}
 > C\sqrt{\frac{\log p}{n}}
 \right]
 \le Cp^{-c}.
\]
Taking the maximum over $a$ and $b$, substituting the resulting bound into \eqref{eq:Ty_expand_2}, and using the bounded inverse Jacobian yields
\[
 \sup_{z\in\cB(r)}\normmax{\widehat{\mSigma}_{\mathrm{Ty}}(z)-\mSigma_{\mathrm{Ty}}(z)}
 \le C\sqrt{\frac{\log p}{n}},
\]
with probability at least $1-Cp^{-c}$. This is \eqref{eq:tme_concentration}.

Apply Lemma~\ref{lem:poet_perturb} to the second POET step. On the same event,
\[
 \sup_{z\in\cB(r)}\normmax{\widehat{\mSigma}_{\mathrm{pt}}(z)-\mSigma_{\mathrm{pt}}(z)}
 \le C\Biggl\{\sqrt{\frac{\log p}{n}}+\rho_{n,p}\Biggr\},
\]
which is \eqref{eq:ptg_concentration}. Finally, Lemma~\ref{lem:glasso_perturb} gives
\[
 \sup_{z\in\cB(r)}\normone{\widehat{\mP}_h(z)-\mP_h(z)}
 \le C s_\Omega\Biggl\{\sqrt{\frac{\log p}{n}}+\rho_{n,p}+\lambda_{\Omega,n}\Biggr\}
 = C\zeta_n.
\]
This proves \eqref{eq:glasso_concentration}.
\end{proof}

\subsection{Proof of Theorem~\ref{thm:main_recursion} and Corollary~\ref{cor:canonical_rate}}

\begin{proof}[Proof of Theorem~\ref{thm:main_recursion}]
For each $t<T_n$ insert the smoothed population map and the oracle map:
\begin{align}
 e_{t+1}
 &= d\bigl(M_{n,h}(z^{(t)}),z^*\bigr) \notag\\
 &\le d\bigl(M_{n,h}(z^{(t)}),M_h(z^{(t)})\bigr)
 + d\bigl(M_h(z^{(t)}),M(z^{(t)})\bigr)
 + d\bigl(M(z^{(t)}),z^*\bigr). 
 \label{eq:main_split}
\end{align}
The first term in \eqref{eq:main_split} is the empirical fluctuation. By Theorem~\ref{thm:sample_nonprecision} and Theorem~\ref{thm:tme_glasso},
\begin{equation}
\label{eq:term_empirical}
 d\bigl(M_{n,h}(z^{(t)}),M_h(z^{(t)})\bigr)
 \le
 C_2\sqrt{\frac{\log n}{nh}}
 + C_3\sqrt{\frac{\log p}{n}}
 + C_4\zeta_n,
\end{equation}
with probability at least $1-Cp^{-c}-Cn^{-c}$, uniformly over $0\le t<T_n$. The second term is the smoothing bias. Proposition~\ref{prop:bias} gives
\begin{equation}
\label{eq:term_bias}
 d\bigl(M_h(z^{(t)}),M(z^{(t)})\bigr)
 \le C_1 h^2.
\end{equation}
The third term is the population contraction. Since the initialization lies in $\cB(r/2)$ and the right-hand side in \eqref{eq:main_recursion} is $o(r)$ by Assumption~\ref{ass:init}, an induction argument shows that every iterate remains in $\cB(r)$. Theorem~\ref{thm:population} therefore yields
\begin{equation}
\label{eq:term_population}
 d\bigl(M(z^{(t)}),z^*\bigr)
 \le \kappa e_t.
\end{equation}
Insert \eqref{eq:term_empirical}, \eqref{eq:term_bias}, and \eqref{eq:term_population} into \eqref{eq:main_split}. This gives \eqref{eq:main_recursion}.

To obtain \eqref{eq:unrolled_rate}, iterate \eqref{eq:main_recursion}:
\begin{align*}
 e_t
 &\le \kappa^t e_0
 + \sum_{u=0}^{t-1}\kappa^u\Biggl(C_1 h^2 + C_2\sqrt{\frac{\log n}{nh}} + C_3\sqrt{\frac{\log p}{n}} + C_4\zeta_n\Biggr) \\
 &\le \kappa^t e_0
 + \frac{C_1 h^2 + C_2\sqrt{\log n/(nh)} + C_3\sqrt{\log p/n} + C_4\zeta_n}{1-\kappa}.
\end{align*}
Finally choose $T_n\asymp \log(1/\varepsilon_{n,h,\zeta})$. Then $\kappa^{T_n}e_0\lesssim \varepsilon_{n,h,\zeta}$ and \eqref{eq:param_consistency} follows. Equation \eqref{eq:param_components} is a direct consequence of the definition of $d$ in \eqref{eq:metric}.
\end{proof}

\begin{proof}[Proof of Corollary~\ref{cor:canonical_rate}]
Under the stated bandwidth choice,
\[
 h^2\asymp \left(\frac{\log n}{n}\right)^{2/5},
 \qquad
 \sqrt{\frac{\log n}{nh}}\asymp \left(\frac{\log n}{n}\right)^{2/5}.
\]
Assumption~\ref{ass:init} gives $\rho_{n,p}\lesssim \sqrt{\log p/n}$ and $\lambda_{\Omega,n}\asymp \sqrt{\log p/n}$. Hence
\[
 \zeta_n
 =
 s_\Omega\Biggl\{\sqrt{\frac{\log p}{n}}+\rho_{n,p}+\lambda_{\Omega,n}\Biggr\}
 \lesssim (1+s_\Omega)\sqrt{\frac{\log p}{n}}.
\]
Substituting these three displays into \eqref{eq:param_consistency} yields \eqref{eq:canonical_rate}.
\end{proof}

\subsection{Proof of Theorem~\ref{thm:risk}}

\begin{proof}[Proof of Theorem~\ref{thm:risk}]
Work on the high-probability event from Theorem~\ref{thm:main_recursion}, on which \eqref{eq:param_components} holds.  Fix $\bx$ and define
\[
 p_k^*(\bx)=\frac{\exp\{\eta_k^*(\bx)\}}{\sum_{\ell=1}^K \exp\{\eta_\ell^*(\bx)\}},
 \qquad
 \widehat p_k(\bx)=\frac{\exp\{\widehat\eta_k(\bx)\}}{\sum_{\ell=1}^K \exp\{\widehat\eta_\ell(\bx)\}}.
\]
If $\widehat G(\bx)\neq G^*(\bx)$, then by definition of the oracle margin,
\[
 \eta_{G^*(\bx)}^*(\bx)-\eta_{\widehat G(\bx)}^*(\bx)
 \le
 \abs{\widehat\eta_{G^*(\bx)}(\bx)-\eta_{G^*(\bx)}^*(\bx)}
 +
 \abs{\widehat\eta_{\widehat G(\bx)}(\bx)-\eta_{\widehat G(\bx)}^*(\bx)}.
\]
Hence
\begin{equation}
\label{eq:margin_reduction}
 \cM^*(\bx)
 \le
 2\max_{1\le k\le K}\abs{\widehat\eta_k(\bx)-\eta_k^*(\bx)}
 \qquad\text{on }\{\widehat G(\bx)\neq G^*(\bx)\}.
\end{equation}
Next,
\begin{align}
 \abs{\widehat\eta_k(\bx)-\eta_k^*(\bx)}
 &\le c_\pi^{-1}\norminf{\widehat{\bpi}-\bpi^*}
 + \sup_{0\le u\le U_0}\abs{\widehat\ell(u)-\ell^*(u)} 
\notag\\
 &\qquad + \sup_{\xi\in[0,U_0]} w_\xi\,\abs{\widehat\delta_k(\bx)-\delta_k^*(\bx)},
 \label{eq:score_diff_1}
\end{align}
where $w_\xi$ denotes the score at an intermediate radius.  By the quadratic expansion used in \eqref{eq:delta_expand},
\[
 \abs{\widehat\delta_k(\bx)-\delta_k^*(\bx)}
 \le C\{1+\norminf{\bx}\}^2\Bigl(\max_{u\le K}\norminf{\widehat{\bmu}_u-\bmu_u^*}+\normone{\widehat\mP-\mP^*}\Bigr).
\]
Conditioning on $Z=k_0$, write $\bx=\bmu_{k_0}^*+\mSigma^{*1/2}R\bm U$.  Then $\norminf{\bx}\le C(1+R)$, and the score decay \eqref{eq:score_decay} yields
\begin{equation}
\label{eq:score_diff_2}
 \max_{1\le k\le K}\abs{\widehat\eta_k(\bm X)-\eta_k^*(\bm X)}
 \le C(1+R)\varepsilon_{n,h,\zeta}
 \qquad\text{almost surely on the event of Theorem~\ref{thm:main_recursion}.}
\end{equation}
Combining \eqref{eq:margin_reduction} and \eqref{eq:score_diff_2} gives
\begin{equation}
\label{eq:mis_event}
 \{\widehat G(\bm X)\neq G^*(\bm X)\}
 \subseteq
 \{\cM^*(\bm X)\le C(1+R)\varepsilon_{n,h,\zeta}\}.
\end{equation}
For every $\bx$, if $a\ge b$, then
\[
 \frac{e^a-e^b}{\sum_{\ell=1}^K e^{\eta_\ell^*(\bx)}}
 \le a-b.
\]
Taking $a=\eta_{G^*(\bx)}^*(\bx)$ and $b=\eta_{\widehat G(\bx)}^*(\bx)$ shows that
\[
 0\le p_{G^*(\bx)}^*(\bx)-p_{\widehat G(\bx)}^*(\bx)
 \le \eta_{G^*(\bx)}^*(\bx)-\eta_{\widehat G(\bx)}^*(\bx)
 \le \cM^*(\bx).
\]
Therefore
\begin{equation}
\label{eq:risk_reduce}
 R(\widehat G)-R(G^*)
 \le \E\Bigl[\cM^*(\bm X)\bone\{\widehat G(\bm X)\neq G^*(\bm X)\}\Bigr].
\end{equation}
Using \eqref{eq:mis_event},
\begin{align*}
 R(\widehat G)-R(G^*)
 &\le \E\Bigl[\cM^*(\bm X)\bone\{\cM^*(\bm X)\le C(1+R)\varepsilon_{n,h,\zeta}\}\Bigr] \\
 &= \E\Bigl[\E\bigl\{\cM^*(\bm X)\bone\{\cM^*(\bm X)\le C(1+R)\varepsilon_{n,h,\zeta}\}\mid R\bigr\}\Bigr].
\end{align*}
For any fixed $R=r$, let $t_r=C(1+r)\varepsilon_{n,h,\zeta}$.  Then
\[
 \cM^*(\bm X)\bone\{\cM^*(\bm X)\le t_r\}
 \le t_r\bone\{0<\cM^*(\bm X)\le t_r\},
\]
and Assumption~\ref{ass:margin} gives
\[
 \E\bigl[\cM^*(\bm X)\bone\{\cM^*(\bm X)\le t_r\}\mid R=r\bigr]
 \le t_r\Prob\{0<\cM^*(\bm X)\le t_r\mid R=r\}
 \le C t_r^{\alpha+1}.
\]
Hence
\[
 R(\widehat G)-R(G^*)
 \le C\varepsilon_{n,h,\zeta}^{\alpha+1}\E(1+R)^{\alpha+1}.
\]
Assumption~\ref{ass:margin} makes the last expectation finite, so
\[
 R(\widehat G)-R(G^*)
 \le C\varepsilon_{n,h,\zeta}^{\alpha+1}
\]
with probability tending to one.  This proves \eqref{eq:risk_rate}.  Setting $\alpha=1$ gives \eqref{eq:risk_rate_alpha1}.
\end{proof}

\section{Additional simulation results}
\label{app:add_sim}

This appendix reports two sensitivity analyses that complement the main fixed-$K$ simulation study in Section~\ref{subsec:simulation_fixed_k}.  The first changes the common scatter matrix from the sparse-precision AR(1) design used in the main text to a dense compound-symmetric design.  The second keeps the AR(1) dependence structure but replaces the sparse mean vectors by dense block means so that every coordinate carries signal.  Together, these experiments separate sensitivity to the scatter structure from sensitivity to the sparsity pattern of the mean signal.  The table notation is the same as in the main simulation section: \emph{Sparse KM} denotes sparse $K$-means, \emph{Sparse KMed} denotes sparse $K$-median, and boldface marks the best feasible (nonoracle) entry within each row.

\subsection{Dense common-scatter sensitivity}
\label{app:cs_sim}

We first assess sensitivity to the common-scatter specification.  The setup is the same as in Section~\ref{subsec:simulation_fixed_k} except for two changes.  The common scatter matrix is replaced by the compound-symmetric matrix
\[
 \mSigma^*_{\mathrm{CS}}
 =
 0.5\,\bone\bone^\top + 0.5\,\mI_p,
\]
so that each diagonal entry equals $1$ and each off-diagonal entry equals $0.5$, and the signal level is increased to $\delta=2$.  The sample size remains $n=300$, the number of clusters is fixed at $K=3$, the dimensions are $p\in\{100,200\}$, and the four radial laws are the Gaussian, $t_5$, Laplace, and slash distributions used in the main text.  This experiment is a dense-scatter sensitivity check rather than a direct verification of the sparse-precision working model.

\begin{table}[H]
\centering
\caption{Additional simulation with compound-symmetric common scatter: clustering accuracy.}
\label{tab:appendix_cs_accuracy}
\scriptsize
\setlength{\tabcolsep}{3.2pt}
\begin{tabular}{lccccccc}
\toprule
\multicolumn{8}{c}{\textit{$p=100$}}\\
\midrule
Distribution & Oracle & $K$-means & Sparse KM & $K$-median & Sparse KMed & CHIME & GEM \\
\midrule
Gaussian & 1.000 & 0.379 & 0.920 & 0.375 & \textbf{0.948} & 0.384 & 0.882 \\
$t_5$ & 0.997 & 0.381 & 0.896 & 0.378 & 0.959 & 0.389 & \textbf{0.970} \\
Laplace & 0.997 & 0.372 & 0.888 & 0.372 & 0.964 & 0.385 & \textbf{0.997} \\
Slash & 0.997 & 0.379 & 0.893 & 0.377 & \textbf{0.977} & 0.389 & 0.939 \\
\midrule
\multicolumn{8}{c}{\textit{$p=200$}}\\
\midrule
Distribution & Oracle & $K$-means & Sparse KM & $K$-median & Sparse KMed & CHIME & GEM \\
\midrule
Gaussian & 1.000 & 0.374 & 0.848 & 0.372 & \textbf{0.897} & 0.374 & 0.871 \\
$t_5$ & 0.996 & 0.375 & 0.827 & 0.375 & 0.966 & 0.379 & \textbf{0.976} \\
Laplace & 0.996 & 0.369 & 0.797 & 0.370 & \textbf{0.971} & 0.373 & 0.962 \\
Slash & 0.998 & 0.371 & 0.771 & 0.370 & \textbf{0.954} & 0.373 & 0.937 \\
\bottomrule
\end{tabular}
\end{table}

\begin{table}[H]
\centering
\caption{Additional simulation with compound-symmetric common scatter: adjusted Rand index.}
\label{tab:appendix_cs_ari}
\scriptsize
\setlength{\tabcolsep}{3.2pt}
\begin{tabular}{lccccccc}
\toprule
\multicolumn{8}{c}{\textit{$p=100$}}\\
\midrule
Distribution & Oracle & $K$-means & Sparse KM & $K$-median & Sparse KMed & CHIME & GEM \\
\midrule
Gaussian & 1.000 & 0.004 & 0.843 & 0.002 & \textbf{0.899} & 0.006 & 0.825 \\
$t_5$ & 0.990 & 0.004 & 0.805 & 0.003 & 0.917 & 0.011 & \textbf{0.954} \\
Laplace & 0.990 & 0.002 & 0.787 & 0.001 & 0.912 & 0.012 & \textbf{0.990} \\
Slash & 0.992 & 0.003 & 0.801 & 0.003 & \textbf{0.942} & 0.008 & 0.906 \\
\midrule
\multicolumn{8}{c}{\textit{$p=200$}}\\
\midrule
Distribution & Oracle & $K$-means & Sparse KM & $K$-median & Sparse KMed & CHIME & GEM \\
\midrule
Gaussian & 0.999 & 0.001 & 0.725 & 0.001 & \textbf{0.816} & 0.002 & 0.785 \\
$t_5$ & 0.989 & 0.002 & 0.695 & 0.002 & 0.921 & 0.004 & \textbf{0.958} \\
Laplace & 0.989 & 0.001 & 0.648 & 0.000 & 0.923 & 0.003 & \textbf{0.934} \\
Slash & 0.994 & 0.000 & 0.612 & 0.000 & \textbf{0.902} & 0.001 & 0.898 \\
\bottomrule
\end{tabular}
\end{table}

The compound-symmetric design changes the ranking among feasible procedures.  The oracle rule remains essentially perfect, but once the common scatter matrix is dense and highly correlated, no single feasible method dominates uniformly across all radial laws.  CHIME is clearly not competitive in this design: its average accuracy remains between $0.373$ and $0.389$, and its average ARI stays close to zero throughout.  By contrast, sparse $K$-median is strongest in the Gaussian and slash settings, whereas GEM is best in the heavier-tailed $t_5$ and Laplace settings and remains close to the best feasible method elsewhere.  Thus the main-text advantage of GEM is attenuated when the scatter side departs strongly from the sparse-precision working model, but the method remains highly competitive under heavy-tailed elliptical components.

\subsection{Dense-mean sensitivity under AR(0.5) dependence}
\label{app:dense_mean_sim}

We next examine a complementary design in which the mean signal is dense rather than sparse.  The setup is the same as in Section~\ref{subsec:simulation_fixed_k} except for the mean structure and signal level.  We keep the AR(1) common scatter matrix with correlation parameter $0.5$, take $p\in\{100,200\}$, and set the signal strength to $\delta=0.4$.  Let $q=p/4$ and define the four consecutive coordinate blocks
\[
 B_r=\{q(r-1)+1,\dots,qr\},\qquad r=1,2,3,4.
\]
With $\bm{1}_{B_r}\in\R^p$ denoting the indicator vector of block $B_r$, the three component centers are
\[
 \bmu_1^*
 =
 \delta\bigl(1.5\bm{1}_{B_1}+0.5\bm{1}_{B_2}-0.5\bm{1}_{B_3}-1.5\bm{1}_{B_4}\bigr),
\]
\[
 \bmu_2^*
 =
 \delta\bigl(-0.5\bm{1}_{B_1}+1.5\bm{1}_{B_2}+0.5\bm{1}_{B_3}-1.5\bm{1}_{B_4}\bigr),
\]
\[
 \bmu_3^*
 =
 \delta\bigl(0.5\bm{1}_{B_1}-1.5\bm{1}_{B_2}+1.5\bm{1}_{B_3}-0.5\bm{1}_{B_4}\bigr).
\]
Thus every coordinate is informative, so the experiment directly examines whether the methods remain effective when feature sparsity is no longer aligned with the data-generating mechanism.  The sample size is $n=300$, the number of clusters is fixed at $K=3$, the four radial laws are the Gaussian, $t_5$, Laplace, and slash distributions used in the main text, and each entry is averaged over $100$ Monte Carlo replications.

\begin{table}[H]
\centering
\caption{Dense-mean sensitivity: clustering accuracy.}
\label{tab:appendix_densemean_accuracy}
\scriptsize
\setlength{\tabcolsep}{3.2pt}
\begin{tabular}{lccccccc}
\toprule
\multicolumn{8}{c}{\textit{$p=100$}}\\
\midrule
Distribution & Oracle & $K$-means & Sparse KM & $K$-median & Sparse KMed & CHIME & GEM \\
\midrule
Gaussian & 0.936 & \textbf{0.918} & 0.800 & 0.896 & 0.603 & 0.896 & 0.888 \\
$t_5$ & 0.945 & 0.850 & 0.804 & 0.920 & 0.772 & 0.842 & \textbf{0.930} \\
Laplace & 0.936 & 0.923 & 0.871 & 0.901 & 0.808 & 0.916 & \textbf{0.924} \\
Slash & 0.944 & 0.820 & 0.766 & 0.913 & 0.686 & 0.811 & \textbf{0.928} \\
\midrule
\multicolumn{8}{c}{\textit{$p=200$}}\\
\midrule
Distribution & Oracle & $K$-means & Sparse KM & $K$-median & Sparse KMed & CHIME & GEM \\
\midrule
Gaussian & 0.985 & \textbf{0.980} & 0.965 & 0.971 & 0.776 & 0.980 & 0.977 \\
$t_5$ & 0.978 & 0.871 & 0.864 & 0.962 & 0.886 & 0.871 & \textbf{0.969} \\
Laplace & 0.972 & \textbf{0.967} & 0.952 & 0.904 & 0.889 & \textbf{0.967} & 0.965 \\
Slash & 0.982 & 0.878 & 0.871 & 0.958 & 0.865 & 0.878 & \textbf{0.976} \\
\bottomrule
\end{tabular}
\end{table}

\begin{table}[H]
\centering
\caption{Dense-mean sensitivity: adjusted Rand index.}
\label{tab:appendix_densemean_ari}
\scriptsize
\setlength{\tabcolsep}{3.2pt}
\begin{tabular}{lccccccc}
\toprule
\multicolumn{8}{c}{\textit{$p=100$}}\\
\midrule
Distribution & Oracle & $K$-means & Sparse KM & $K$-median & Sparse KMed & CHIME & GEM \\
\midrule
Gaussian & 0.820 & \textbf{0.773} & 0.529 & 0.716 & 0.255 & 0.718 & 0.707 \\
$t_5$ & 0.843 & 0.717 & 0.607 & 0.781 & 0.516 & 0.698 & \textbf{0.804} \\
Laplace & 0.818 & 0.784 & 0.662 & 0.745 & 0.580 & 0.766 & \textbf{0.787} \\
Slash & 0.840 & 0.683 & 0.547 & 0.766 & 0.375 & 0.655 & \textbf{0.798} \\
\midrule
\multicolumn{8}{c}{\textit{$p=200$}}\\
\midrule
Distribution & Oracle & $K$-means & Sparse KM & $K$-median & Sparse KMed & CHIME & GEM \\
\midrule
Gaussian & 0.954 & \textbf{0.943} & 0.900 & 0.916 & 0.549 & 0.942 & 0.933 \\
$t_5$ & 0.935 & 0.798 & 0.776 & 0.898 & 0.775 & 0.798 & \textbf{0.917} \\
Laplace & 0.917 & 0.903 & 0.866 & 0.820 & 0.763 & 0.903 & \textbf{0.904} \\
Slash & 0.948 & 0.814 & 0.790 & 0.899 & 0.723 & 0.814 & \textbf{0.930} \\
\bottomrule
\end{tabular}
\end{table}

The dense-mean experiment changes the comparison in a direction different from the sparse-mean setting in the main text.  Because all coordinates carry signal, sparse feature-selection methods are no longer well matched to the design; this is most visible for sparse $K$-median, whose selected subsets discard part of the dense separation and lead to substantially smaller accuracy and ARI.  When $p=100$, GEM is the best feasible method under the three non-Gaussian radial laws, whereas ordinary $K$-means is strongest in the Gaussian case.  When $p=200$, the dense signal accumulates over more coordinates and the Gaussian and Laplace cases become easier for the Euclidean and CHIME-type procedures; nevertheless, GEM remains the best feasible method under the $t_5$ and slash laws and has the best ARI in the Laplace case.  These results complement the sparse-mean experiments: the advantage of sparsity-aware initialization is reduced when there are no inactive coordinates, but the semiparametric elliptical update remains useful for heavy-tailed radial behavior.

\bibliographystyle{apa}
\bibliography{ref}

\end{document}